\documentclass[sigconf]{acmart}

\AtBeginDocument{%
	}

\newcommand{\name}[0]{Gus}

\usepackage{svg,tcolorbox}
\usepackage{balance}
\usepackage{enumitem}
\usepackage{comment}
\usepackage[normalem]{ulem}
\usepackage{url}
\usepackage{algorithm,amsmath,amsthm,amsfonts}%,amssymb}
\usepackage{thmtools,thm-restate}
\usepackage{multicol}
\usepackage{hyperref}
\usepackage{lineno}
\usepackage[framemethod=TikZ]{mdframed}
\usepackage[noend]{algorithmic}

\usepackage{graphicx}
\newtheorem{theorem}{Theorem}
\newtheorem{lemma}{Lemma}
\newtheorem{definition}{Definition}

\newcommand{\remove}[1]{}

\usepackage{subcaption,graphicx}
\usepackage{tikz}
\usepackage{amsmath}

\copyrightyear{2023}
\acmYear{2023}
\setcopyright{rightsretained}
\acmConference[SPAA '23]{Proceedings of the 35th ACM Symposium on Parallelism in Algorithms and Architectures}{June 17--19, 2023}{Orlando, FL, USA}
\acmBooktitle{Proceedings of the 35th ACM Symposium on Parallelism in Algorithms and Architectures (SPAA '23), June 17--19, 2023, Orlando, FL, USA}
\acmDOI{10.1145/3558481.3591086}
\acmISBN{978-1-4503-9545-8/23/06}

\begin{document}
	
	%%
	%% The "title" command has an optional parameter,
	%% allowing the author to define a "short title" to be used in page headers.
	\title{Distributed Multi-writer Multi-reader Atomic Register with Optimistically Fast Read and Write}
	
	%%
	%% The "author" command and its associated commands are used to define
	%% the authors and their affiliations.
	%% Of note is the shared affiliation of the first two authors, and the
	%% "authornote" and "authornotemark" commands
	%% used to denote shared contribution to the research.
	\author{Lewis Tseng}
	\email{lewis.tseng@bc.edu}
	\affiliation{%
		\institution{Boston College}
		\city{Boston}
		\state{MA}
		\country{USA}
	}
	
	\author{Neo Zhou}
	\email{zhouaea@bc.edu}
	\affiliation{%
		\institution{Boston College}
		\city{Boston}
		\state{MA}
		\country{USA}
	}
	
	\author{Cole Dumas}
	\email{dumasca@bc.edu}
	\affiliation{%
		\institution{Boston College}
		\city{Boston}
		\state{MA}
		\country{USA}
	}
	
	\author{Tigran Bantikyan}
	\email{tigranbantikyan@u.northwestern.edu}
	\affiliation{%
		\institution{Northwestern}
		\city{Evanston}
		\state{IL}
		\country{USA}}
	
	\author{Roberto Palmieri}
	\email{palmieri@lehigh.edu}
	\affiliation{%
		\institution{Lehigh University}
		\city{Bethlehem}
		\state{PA}
		\country{USA}}
	
	%%
	%% By default, the full list of authors will be used in the page
	%% headers. Often, this list is too long, and will overlap
	%% other information printed in the page headers. This command allows
	%% the author to define a more concise list
	%% of authors' names for this purpose.
	\renewcommand{\shortauthors}{Tseng et al.}
	
	%%
	%% The abstract is a short summary of the work to be presented in the
	%% article.
	\begin{abstract}
		A distributed multi-writer multi-reader (MWMR) atomic register is an important primitive that enables a wide range of distributed algorithms. Hence, improving its performance can have large-scale consequences. Since the seminal work of ABD emulation in the message-passing networks [JACM '95], many researchers study fast implementations of atomic registers under various conditions. ``Fast'' means that a read or a write can be completed with 1 round-trip time (RTT), by contacting a simple majority.  
		In this work, we explore an atomic register with optimal resilience and ``\textit{optimistically fast}'' read and write operations. That is, both operations can be fast if there is \textit{no} concurrent write. 
		
		This paper has three contributions: (i) We present \name, the emulation of an MWMR atomic register with optimal resilience and optimistically fast reads and writes when there are up to 5 nodes; (ii) We show that when there are $> 5$ nodes, it is \textit{impossible} to emulate an MWMR atomic register with  both properties; and (iii) We implement \name{} in the framework of EPaxos and Gryff, and show that \name{} provides lower tail latency than state-of-the-art systems such as EPaxos, Gryff, Giza, and \textsc{Tempo} under various workloads in the context of geo-replicated object storage systems.
	\end{abstract}
	
	
	%%
	%% The code below is generated by the tool at http://dl.acm.org/ccs.cfm.
	%% Please copy and paste the code instead of the example below.
	%%
	% \begin{CCSXML}
	% <ccs2012>
	%  <concept>
	%   <concept_id>10010520.10010553.10010562</concept_id>
	%   <concept_desc>Computer systems organization~Embedded systems</concept_desc>
	%   <concept_significance>500</concept_significance>
	%  </concept>
	%  <concept>
	%   <concept_id>10010520.10010575.10010755</concept_id>
	%   <concept_desc>Computer systems organization~Redundancy</concept_desc>
	%   <concept_significance>300</concept_significance>
	%  </concept>
	%  <concept>
	%   <concept_id>10010520.10010553.10010554</concept_id>
	%   <concept_desc>Computer systems organization~Robotics</concept_desc>
	%   <concept_significance>100</concept_significance>
	%  </concept>
	%  <concept>
	%   <concept_id>10003033.10003083.10003095</concept_id>
	%   <concept_desc>Networks~Network reliability</concept_desc>
	%   <concept_significance>100</concept_significance>
	%  </concept>
	% </ccs2012>
	% \end{CCSXML}
	
	% \ccsdesc[500]{Computer systems organization~Embedded systems}
	% \ccsdesc[300]{Computer systems organization~Redundancy}
	% \ccsdesc{Computer systems organization~Robotics}
	% \ccsdesc[100]{Networks~Network reliability}
	
	%%
	%% Keywords. The author(s) should pick words that accurately describe
	%% the work being presented. Separate the keywords with commas.
	%% A "teaser" image appears between the author and affiliation
	%% information and the body of the document, and typically spans the
	%% page.
	% \begin{teaserfigure}
	%   \includegraphics[width=\textwidth]{sampleteaser}
	%   \caption{Seattle Mariners at Spring Training, 2010.}
	%   \Description{Enjoying the baseball game from the third-base
	%   seats. Ichiro Suzuki preparing to bat.}
	%   \label{fig:teaser}
	% \end{teaserfigure}
	
	% \received{20 February 2007}
	% \received[revised]{12 March 2009}
	% \received[accepted]{5 June 2009}
	
	%%
	%% This command processes the author and affiliation and title
	%% information and builds the first part of the formatted document.
	\maketitle
	
	\section{Introduction}
\label{s:intro}

Attiya, Bar-Noy, Dolev \cite{AttiyaBD1995}  present an emulation algorithm, namely ABD, that implements an atomic single-writer multi-reader register with optimal resilience in asynchronous message-passing networks when nodes may crash. ABD allows porting many known shared-memory algorithms to message-passing networks, such as multi-writer multi-reader (MWMR) registers, atomic snapshots, approximate consensus and randomized consensus. 

The MWMR version of ABD \cite{GeoQuorum_Lynch1997} requires 2 RTT to complete a write operation, and 1 RTT to complete a read when there is no concurrent write. Subsequent works identify conditions so that reads \cite{ABD_read_Rachid_PODC2004,ABD_Analysis_PODC2020} and writes \cite{ABD_write_OPODIS2009} can be fast. An operation is fast if it can \underline{always} be completed in 1 round-trip time (RTT), by contacting a simple majority of nodes.  Unfortunately, the conditions for fast writes are not generally applicable to practical systems as will be discussed in Section \ref{s:related}. 

Dutta et al. [PODC '04] prove that  implementing an atomic register with \textit{both} fast writes and fast reads is impossible \cite{ABD_read_Rachid_PODC2004}. Recently, Huang et al. [PODC '20] identify more constraints in implementing fast writes or fast reads \cite{ABD_Analysis_PODC2020}. Motivated by these results, we ask: ``\textit{Can we do better for practical systems?}''

%\paragraph{Motivation}
\vspace{3pt}
\noindent\textbf{Motivation.}~
Observe that  object storage systems can be modeled as atomic registers. For real-world object storages, the typical workloads have two key characteristics \cite{Giza_ATC2017,IBM-Cloud-Object_PDSW2015,IBM-Cloud-Storage_HotStorage2020}: (i) \textit{Concurrency is rare, but possible}: in Microsoft OneDrive, only 0.5\% of the writes occur within a 1 second interval; and (ii) \textit{Object size and operation vary widely}: IBM Cloud Object Storage supports hosting service of web page, game, video and enterprise backups. In  their testing benchmark \cite{IBM-Cloud-Object_PDSW2015}, object size varies from 1 KB to 128 MB, and the ratio of write operations range from 5\% to 90\%. 

These observations indicate that it is important to design an algorithm that handles various workloads efficiently, for practical object storages. To optimize for the common case, we are interested in ``\textit{optimistically fast}'' operations, i.e., operations that are fast, when there is \underline{no concurrent write}. The MWMR version of ABD \cite{GeoQuorum_Lynch1997} achieves optimistically fast reads, but not writes. Concretely, we answer the following question in this paper:

\begin{tcolorbox}
\centering
Is it possible to implement an atomic register that supports both \textit{optimistically fast} reads and writes?

\end{tcolorbox}

\vspace{3pt}
\noindent\textbf{Contribution: Theory.}~
On the positive side, we present \name, which implements an MWMR atomic register with optimal resilience and optimistically fast read and write operations when there are up to 5 nodes. To achieve optimistically fast operations, \name{} combines two novel techniques: (i) \textit{speculative timestamp}: a node optimistically uses locally known logical timestamp to enable 1-RTT fast path for writes  (i.e., when writes commit in a single communication step), and (ii) \textit{view exchange}: nodes exchange newly received timestamps to enable 1-RTT fast path for reads. 

Considering that most production storage systems deploy 3- or 5-way geo-replication \cite{GoogleFileSystem_SOSP2003,Spanner_OSDI12,Pileus_SOSP2013}, we believe \name{} is useful for practical settings, given its performance benefits. Furthermore, to address the scalability issue, we
propose two solutions with  different trade-offs between latency (in terms of RTTs) and resilience.

Furthermore, we formally prove that scalability is fundamentally limited. We show that when there are $> 5$ nodes, it is \textit{impossible} to emulate an optimally resilient atomic register that supports optimistically fast reads and writes. This impossibility implies that \name{} is optimal with this aspect.

\vspace{3pt}
\noindent\textbf{Contribution: Systems and Experiments.}~
We experimentally evaluate how the property of optimistically fast operations behave in object storages, as it is difficult to quantify how concurrent operations affect the performance in theory.  Practical systems often use a consensus-based approach to implement an object storage. Hence, we compare \name{} with state-of-art consensus-based systems EPaxos \cite{EPaxos_SOSP13}, Gryff \cite{Gryff_Lloyd_NSDI2020},Giza \cite{Giza_ATC2017}, and \textsc{Tempo} \cite{Tempo_Alexey_Eurosys21}.

We implemented \name{} in the framework of EPaxos \cite{EPaxos_SOSP13} and Gryff \cite{Gryff_Lloyd_NSDI2020} to make a fair comparison. Furthermore, in the same framework, we implemented our version of Giza \cite{Giza_ATC2017} (source code not available). \name{} outperforms these competitors in both throughput and latency, which demonstrates  practical performance benefits under a wide range of workloads. Under various settings with three nodes, \name{} has better tail latency than both Gryff and EPaxos. Compared to Gryff, 5\%-18\% of \name's reads are faster, and $\geq$95\% of writes improve latency by up to 50\%. \name{} also has 0.5x-4.5x maximum throughput than both Gryff and EPaxos in the case of write-intensive geo-replication workload. With 9 nodes, \name{}'s tail latency for reads has 12.5\% improvement over \textsc{Tempo}'s \cite{Tempo_Alexey_Eurosys21}.

\section{Preliminaries and Related Work}

\subsection{System Model}
We consider an asynchronous message-passing network consisting of $n$ nodes, where $n \leq 5$. Section \ref{s:optimization-scalability} presents solutions to scale \name{} beyond $5$ nodes with different trade-offs. At most $f$ of the nodes may crash at any point of time. \name{} ensures safety and liveness as long as $n\geq 2f+1$. Messages could be arbitrarily delayed, but messages between any pair of fault-free nodes are delivered eventually. 

Following the convention of the literature \cite{Lynch96,attiya2004distributed,AttiyaBD1995}, we assume that each node has client threads (reader thread or writer thread) and a server thread. In practical systems, this model captures co-located clients -- a client $C$ is co-located with a sever $R$ if the message delivery latency between $C$ and $R$ is much less than the minimum latency between $C$ and other servers. Clients running the applications (e.g., web hosting or backup service) can be considered co-located with a server in the same data center.

\vspace{3pt}
\noindent\textbf{Linearizability.}~\name{} achieves linearizability \cite{herlihy1990linearizability}. That is,  there exists a total ordering of operations $O$ such that (i) operations appear to execute in the order of $O$; (ii) $O$ is consistent with register semantics, i.e., any read must return the value of the most recent write in $O$; and (iii) $O$ satisfies the real-time ordering between operations, i.e., if operation $o_a$ completes before the invocation of operation $o_b$, then $o_b$ should appear after $o_a$ in $O$.

\subsection{Related Work}
\label{s:related}

This section discusses the closely related theory works. We defer the comparison between \name{} and relevant practical systems to Section \ref{s:evaluation}. These systems (e.g., \cite{EPaxos_SOSP13,Gryff_Lloyd_NSDI2020,Giza_ATC2017,Tempo_Alexey_Eurosys21}) are based on some form of consensus and provide liveness only in partially asynchronous networks, whereas \name{} uses quorum and ensures both safety and liveness in asynchronous networks. 

The ABD algorithm by Attiya, Bar-Noy, Dolev \cite{AttiyaBD1995} is the first implementation of atomic single-writer multi-reader (SWMR) register in asynchronous networks with $n \geq 2f+1$. ABD requires 1 RTT for writes and 2 RTT for reads. Lynch and Shvartsman \cite{GeoQuorum_Lynch1997} later extend the algorithm to the multi-writer multi-reader version, which takes 2 RTT for writes and 2 RTT for reads. These two versions of ABD support a simple optimization to make reads \textit{optimistically fast}, i.e., 1 RTT reads when there is \textit{no} concurrent write. 

Subsequent works \cite{ABD_read_Rachid_PODC2004,ABD_Analysis_PODC2020,ABD_write_OPODIS2009} study fast operations which complete in 1 RTT. The algorithm in \cite{ABD_write_OPODIS2009} supports fast writes only when there are at most $n-1$ writer clients. In practical geo-replication with $n$ data centers, this condition implies that one data center cannot serve any writes. 
The algorithms in \cite{ABD_read_Rachid_PODC2004,ABD_Analysis_PODC2020} support fast reads, but require $n = O(fn_R)$, where $n_R$ is the number of readers. 

Prior works identify several impossibilities. Dutta et al. \cite{ABD_read_Rachid_PODC2004} show that in general it is impossible to have both fast writes and reads, when a single node may crash. Englert et al. \cite{ABD_write_OPODIS2009} prove that to support fast writes, the number of writes cannot be more than $n-1$ (which implies that their algorithm is optimal in this aspect). Huang et al. \cite{ABD_Analysis_PODC2020} derive two more impossibilities: (i) fast write is impossible if reads need to be completed in 2 RTT; and (ii) to have fast reads and 2-RTT writes, $\Omega(fn_R)$ is the lower bound on $n$. 

Several works study other variations of the properties, e.g., semi-fast operations \cite{Chryssis_JPDC_atomicity09,Tseng_ICDCS20_BSR}, fast operations for Byzantine-tolerant SWMR registers \cite{Rachid_RQS_DC10}, weak semi-fast operations \cite{Chryssis_atomicity_DISC08}, and fast operations for regular and safe registers \cite{Rachid_fast_regular_PODC06,AbrahamCKM2006}. To the best of our knowledge, no prior work studies the feasibility of atomic registers with \textit{optimistically fast} operations. Furthermore, our idea of speculative timestamp is new, which would be useful for future works that aim to achieve optimistically fast operations.

\vspace{3pt}
\noindent\textbf{ABD Register \cite{AttiyaBD1995,GeoQuorum_Lynch1997}.}~
Most works on atomic registers in message-passing networks are inspired by ABD, including \name{}. Hence, we briefly describe ABD before presenting our design. We describe ABD and \name{} for a single register.  Recall that linearizability is a local (or composable) property \cite{herlihy1990linearizability}, i.e.,  the property holds for a set of objects, \textit{if and only if} it holds for each individual object. Therefore, it is straightforward to compose instances of these protocols to obtain a linearizable system that supports multiple registers. 

ABD associates a unique \textit{tag} with a write and its value. Writes and values are ordered lexicographically by their tags. Formally, a tag is a tuple, $(ts, id)$, consisting of two fields: (i) a logical timestamp representing the (logical) time for the write; and (ii) the writer ID representing the writer client's identifier that invokes the write. For tag $t$, we use ``$t.ts$'' to denote the timestamp field, and ``$t.id$'' to denote the writer ID field. Two tags can be compared as follows: 

\begin{definition} Tag $t_1$ is greater than tag $t_2$ if (i) $t_1.ts > t_2.ts$; or (ii) $t_1.ts = t_2.ts$ and $t_1.id > t_2.id$.

Tag $t_1$ is equal to $t_2$ if $t_1.ts = t_2.ts$ and $t_1.id = t_2.id$.
\end{definition}

Each node stores a value $v$ and an associated tag $t$. ABD register requires two phases for both reads and writes. A read begins with the reader client obtaining the current tag and value from a quorum. The quorum is any simple majority of nodes. The reader then chooses the value associated with the maximum tag and propagates this maximum tag and value to all the nodes. Upon the acknowledgments from a quorum, the read is complete. The second phase, namely the ``\textit{write-back}'' phase, can be omitted if all the tags from the first phase are identical,  achieving \textit{optimistically fast} reads.

A writer client $w$ follows a similar two-phase procedure. It first obtains the maximum tag $t_{max}$ from a quorum, and then  constructs a new tag $t = (t_{max}.ts+1, w)$. In the second phase, client $w$ propagates $t$ and its value to all nodes and waits for acknowledgments from a quorum. Since a writer needs to contact a quorum to obtain tag $t$ (\textit{writer-reads} design), ABD and relevant protocols \cite{ABD_Analysis_PODC2020,ABD_read_Rachid_PODC2004} require 2 RTT for the write operations, even if there is no concurrent operation. Our technique ``speculative timestamp'' and the focus on only $3$ or $5$ nodes allow us to skip this step optimistically.

	\newcommand{\request}[0]{\textsc{Request}}
\newcommand{\view}[0]{View}
\newcommand{\storage}[0]{Storage}
\newcommand{\tmpStorage}[0]{TmpStorage}
\newcommand{\currentTime}[0]{ts^{max}}
\newcommand{\currentTag}[0]{tag}
\newcommand{\currentValue}[0]{value}
\newcommand{\ifStaleTag}[0]{isStale}

\newcommand{\writeData}[0]{\textsc{write}}
\newcommand{\commitWrite}[0]{\textsc{commit-write}}
\newcommand{\ackWrite}[0]{\textsc{ack-write}}
\newcommand{\updateView}[0]{\textsc{update-view}}
\newcommand{\ackCommit}[0]{\textsc{ack-commit}}
\newcommand{\readData}[0]{\textsc{read}}
\newcommand{\ackRead}[0]{\textsc{ack-read}}

\section{\name{}: Design}
\label{s:design}

\subsection{Architecture and  Protocol}
\label{s:protocol}
\name{} borrows tag and lexicographical ordering from ABD. A key challenge is to determine a tag for each write. Later we will show that even with a speculative timestamp, each write and its value still obtain a unique tag. As a result, we will often refer to a tag as the ``version'' of the register value. 

Recall that each node has a writer, a reader and a server.\footnote{Nodes can support multiple writers and readers using proxies.} Writers and readers communicate with server threads at other nodes. For brevity, we will simply say writer/reader communicate with nodes. Readers exchange $\langle \readData{} \rangle$ and $\langle\ackRead{}\rangle$ messages, and writers exchange $\langle \writeData\rangle$, $\langle\ackWrite\rangle$, and $\langle\commitWrite\rangle$ messages. Background handlers of the server implement a set of event-driven functions that exchange $\langle\updateView{}\rangle$ messages with other nodes and update local variables.

\vspace{3pt}
\noindent\textbf{Node States.}~~
Each node $R_i$ maintains three states: 
\begin{itemize}
    \item $\storage_i$ is a set of tuples $(tag, value)$, which stores all the versions of the register value, where each version has a unique $tag$;

    \item $tag_i$ represents the largest known tag associated with the value in $\storage_i$; and

    \item $View_i$ is a vector that keeps track of each node's view. View of a node $R_i$ is defined as a set of tags that $R_i$ has known so far. By design, $View_i[i] $ contains the tags associated to all the values in $\storage_i$. Condition \textsc{SafeToReturn} presented later in Definition \ref{def:safetoread} shows how \name{} uses $View_i$ to decide which version of the register value is safe to return, with respect to linearizability.
\end{itemize}
%(i) $\storage_i$ stores all the versions of the register value, where each version has a unique $tag$; (ii) $tag_i$ represents the largest tag associated with the value in $\storage_i$; and (iii) $View_i$ is a vector that keeps track of each node's view. View of a node $R_i$ is defined as a set of tags that $R_i$ has known so far. By design, $View_i[i] $ contains the tags associated to all the values in $\storage_i$. Condition \textsc{SafeToReturn} presented later in Definition \ref{def:safetoread} shows how \name{} uses $View_i$ to decide which version of the register value is safe to return, with respect to linearizability. 
We assume any thread on the node can access these states. This assumption is typical in many practical systems, as clients are handled by client proxies that run on each node. 

% \begin{figure}[t]
%     \begin{center}
%     \begin{mdframed}[roundcorner=10pt]%
%        \parbox{\textwidth}{%
%             $\storage_i$: ~~a set of tuples $(tag, value)$
    
%             %\vspace{1pt}
            
%             %$seq_i$: current sequence, initially $0$
            
%             \vspace{1pt}
    
%             $\currentTag_i$: ~~~~~~~~~largest known tag in $\storage_i$  $(ts, j)$
            
%             \vspace{1pt}
            
%             $\view_i$: ~~~~~~a vector of other node's received tags
            
%             %$\currentValue_i$: data associated with $\currentTag_i$ in $\storage_i$
%         }
%     \end{mdframed}
%     \end{center}
%     %\vspace{-20pt}
%     \caption{States at node $i$.}
%     \label{fig:state}
% \end{figure}

\vspace{3pt}
\noindent\textbf{Techniques and Challenges.}~
\name{} has two novel techniques: 

\begin{itemize}%[noitemsep,nolistsep]
    \item \textit{Speculative timestamp}: Writer opportunistically uses the local tag $tag_i$ %(the state $tag_i$)
    %as shown in Figure \ref{fig:state}) 
    as the tag for the value it intends to write.
    
    %\name{} uses ``\textit{writers-reads-if-necessary}'' to skip the first phase of ABD (reading the most recent tag) in the contention-free case.
    
    \item \textit{View exchange}: Each node propagates to all the other nodes whenever it has learned a new tag. Each node $i$ uses $View_i$ to keep track of this information. 
\end{itemize}

\textit{Speculative timestamp} allows \name{} to achieve 1 RTT when there is no concurrent write, and enters the second phase only when observing a concurrent write. \textit{View exchange} allows nodes to collect up-to-date information and to enable 1-RTT read when there is no concurrent write. 
%\textit{View exchange} serves two purposes: 
%(i) allowing each node to learn new tags so that it is more likely for writers to have the most up-to-date tag; ---- this can be implemented???
%(i) skipping the second phase of ABD's read when there is no concurrent write; and (ii) supporting sequentially consistent reads that do not require any communication, as will be discussed in Section \ref{s:optimization}. 
In terms of protocol design, we need to address the following two technical challenges:

\begin{itemize}%[noitemsep,nolistsep]
    \item No read can return stale value, even if the speculative timestamp is stale. A writer can observe a stale timestamp if the node that the writer is co-located with has \textit{not} received the most recent writes from other nodes.
    
    \item No write operation can be associated with two tags. Essentially, the ordering of the operations is constructed using the tags; hence, if a write can be associated with multiple tags, the total ordering could be violated. We will formally define what ``associated tags'' mean after presenting the protocol. 
\end{itemize}

\begin{algorithm*}[!hptb]
\begin{algorithmic}[1]
	%\small
% 	\item[{\textbf{NOTE}:}]{}
% 	\item[] Only one write or read can be executed at a time.
% 	\item[] This algorithm only works for $N=3$.
% 	\item[{\bf Local Variables}:]{}
%     \item[] $\view_i$ \COMMENT{a vector of other node's history}
%     \item[] $\storage_i$ \COMMENT{a list of tuples in the form of $((seq, w), time, value)$}
%     \item[] $seq_i$ \COMMENT{current sequence, initially $0$}
%     %\item[] $\currentTime_i$ \COMMENT{known latest logical timestamp}
%     \item[] $\currentTag_i$ \COMMENT{largest known tag in $\storage_i$, $(time, j)$}
%     \item[] $\currentValue_i$ \COMMENT{data associated with $\currentTag_i$ in $\storage_i$}
    
%     \item[] \hrulefill
\vspace{-10pt}
\begin{multicols}{2}
\footnotesize
\item[{\bf Write($value$) invoked by Writer} $i$:]
\STATE // \textit{Put phase: propagate data with speculative timestamp}
\STATE $\currentTime_i \gets \currentTag_i.ts + 1$
\STATE \textbf{send} $\langle \writeData, (\currentTime_i, i), value \rangle$ \textbf{to all}
\STATE \textbf{wait until} receiving $\langle \ackWrite, tag_j \rangle$ \textbf{from a quorum of nodes} $Q$
%\STATE \textbf{wait until} receiving $\langle \ackWrite, \ifStaleTag_j, tag_j \rangle$ \\~~~~~~from a quorum of nodes $Q$\label{line:write-data-quorum}

\vspace{1pt}

\STATE // \textit{Update phase: update tag if necessary}
%\IF{$\ifStaleTag_j = False$ for all $j \in Q$}
\IF{$(\currentTime_i, i) > tag_j$ for all $j \in Q$}
    \STATE commit write  \COMMENT{fast-path for write}
    %\STATE \textbf{send} $\langle \commitWrite, (\currentTime_i, i) \rangle$ \textbf{to all}
    % This can be optimized by sending to nodes that are not in $Q$
\ELSE
    \STATE $\currentTime_i \gets \max_{j \in Q}~ tag_j.ts + 1$
    \STATE \textbf{send} $\langle \commitWrite, (\currentTime_i, i)  \rangle$ \textbf{to all}
    \STATE \textbf{wait until} receiving $\langle \ackCommit, (\currentTime_i, i) \rangle$ from a quorum
\ENDIF
\STATE $\currentTag_i.time \gets \currentTime_i$;~~~~~$\currentTag_i.id \gets i$
\STATE $\storage_i \gets \storage_i \cup \{(\currentTag_i, value)\}$  
\STATE $\view_i[i] \gets \view_i[i] \cup \{(\currentTag_i, i)\}$
%\STATE \textbf{send} $\langle \updateView, (\currentTag_i, i) \rangle$ \textbf{to all}
\STATE commit write, if not have already done so \COMMENT{slow-path for write}

\vspace{3pt}

\item[{\bf Read() invoked by Reader} $i$:]
\STATE \textbf{send} $\langle \readData, i \rangle$ \textbf{to all}
\STATE \textbf{wait until} receiving $\langle \ackRead, tag_j \rangle$ \textbf{from a quorum of nodes} $Q$
\STATE $tag^{max} \gets$ largest $tag_j$ received from all $j \in Q$
%\STATE // The quorum in the next line needs to include $i$ itself
\STATE \textbf{wait until} Condition \textsc{SafeToRead} holds on $\view_i, tag^{max}, value$

% there is a $(tag, value)$ such that \\~~~~~~ (i) $tag \geq tag^{max}$
% \\~~~~~~ (ii) there is a quorum $Q_R$ s.t. $i \in Q_R$ and for each $j \in Q_R, (tag, j) \in \view_i[j]$

%$tag$ is in $\view_i$ for a quorum and $\view_i[i]$
\STATE return $value$

%\STATE return $(time_{max}, value_{max})$
%\item[]

\vspace{3pt}

\item[\textbf{Background Handlers at Node} $R_i$:]
\STATE \textbf{Upon receiving} $(\writeData, tag_j, value)$ \textbf{from writer} $j$:
\IF{$\currentTag_i < tag_j$}
    \STATE $\storage_i \gets \storage_i \cup \{(tag_j, value)\}$ %\COMMENT{Commit $j$'s data}
    \STATE $\currentTag_i \gets tag_j$
    \STATE $\view_i[j] \gets \view_i[j] \cup \{tag_j\}$; ~~~~~$\view_i[i] \gets \view_i[i] \cup \{tag_j\}$
    \STATE \textbf{send} $\langle \updateView, tag_j\rangle$ \textbf{to all nodes}
\ELSE
    \STATE $\tmpStorage_i \gets \tmpStorage_i \cup \{(tag_j, value)\}$

\ENDIF
\STATE \textbf{send} $\langle \ackWrite, \currentTag_i \rangle$ \textbf{to writer} $j$

\vspace{3pt}

\STATE \textbf{Upon receiving }$(\commitWrite, tag_j)$ \textbf{from writer} $j$:

\IF{$j$'s write is in $\tmpStorage_i$}
    \STATE $value \gets $ value associated with $j$'s most recent write in $\tmpStorage_i$
    \STATE $\storage_i \gets \storage_i \cup \{(tag_j, value)\}$   
\ELSE
    \STATE Update $\storage_i$ to ensure $j$'s write has tag $tag_j$
\ENDIF
\IF{$tag_i < tag_j$}
    \STATE $\currentTag_i \gets tag_j$
\ENDIF
\STATE $\view_i[j] \gets \view_i[j] \cup \{tag_j\}$; ~~~~~$\view_i[i] \gets \view_i[i] \cup \{tag_j\}$
\STATE \textbf{send} $\langle \updateView, tag_j \rangle$ \textbf{to all nodes}
\STATE \textbf{send} $\langle \ackCommit, tag_j \rangle$ \textbf{to writer} $j$

\vspace{3pt}

\STATE \textbf{Upon receiving} $(\readData, j)$ \textbf{from reader} $j$:
\STATE \textbf{send} $\langle \ackRead, \currentTag_i \rangle$ \textbf{to reader} $j$

\vspace{3pt}

\STATE \textbf{Upon receiving} $(\updateView, (ts, k))$ \textbf{from node} $j$:
\STATE $\view_i[j] \gets \view_i[j] \cup \{(ts, k)\}$

\end{multicols}
\vspace{-10pt}
\end{algorithmic}
\caption{\name{} for $n=3$}
\label{algo:Gus}
\end{algorithm*}

%\paragraph*{Protocol Specification}
\noindent\textbf{Protocol Specification.}~
Algorithm \ref{algo:Gus} specifies the steps that need to be followed by each node when $n=3$. We defer the discussion of extension of $n=4$ or $5$ to Section \ref{s:n=5}. 

\noindent\textbf{Write operation}: 
Writer $i$, which is co-located with node $R_i$, obtains tag $(ts_i^{max}, i)$ by adding $1$ to the timestamp of the largest tag known to node $R_i$ (Line 2). It then propagates the value along with this new tag to all the nodes and waits until receiving an acknowledgement from a quorum of nodes $Q$ (Line 4). A quorum used in \name{} is always a \textit{simple majority}.

\textit{Fast Path}:~~Writer can then detect whether there is a concurrent write by comparing $(\currentTime_i, i)$ with the tag received from $Q$ (Line 6). If there is no concurrent write operation, then $i$'s write is on the fast path (Line 7). Client is notified that the write is completed at this point. The writer proceeds to asynchronous bookkeeping steps, including updating tag (Line 12),  storage (Line 13), and view (Line 14). All these steps can be done asynchronously, because after Line 7, it is guaranteed that enough nodes have already obtained the value with the correct tag. %Section \ref{s:optimization} presents how to improve message complexity at Line 8. 

\textit{Slow Path}:~~Only if the writer detects a concurrent write, it needs to obtain and update the correct logical timestamp. It first
constructs the logical timestamp by finding the largest timestamp field in the received tags from $Q$ and increasing it by 1 (Line 9). The writer then sends the commit message $\langle \commitWrite{} \rangle$ to all the nodes to update the tag, and waits for acknowledgement from a quorum on the slow path. This is necessary to ensure that enough nodes have received the correct and updated tag. Note that $\langle \commitWrite{} \rangle$ message does \textit{not} include the value field to save network bandwidth.

\textit{Background Handler for Writes}:~~The server thread has event-driven handlers that run in the background to process incoming messages. Upon receiving $\langle \writeData{} \rangle$ message from writer $j$, node $R_i$ first checks the tag $(ts, j)$. If it is larger than $\currentTag_i$, then node $R_i$ stores the value (Line 23), updates tag (Line 24) and view (Line 25), and notifies others that it has learned a new value (Line 26). Finally, $R_i$ replies $j$ with the acknowledgement (Line 29). If $tag_i$ is larger, this means that writer $j$'s tag may be stale, and $j$ needs to update its speculative timestamp later. Hence, node $R_i$ puts the value at a temporary storage $\tmpStorage_i$ (Line 28) and replies $j$ (Line 29). %Section \ref{s:optimization} discusses trade-off at this step. 

Upon receiving $\langle \commitWrite{} \rangle$ message from writer $j$, node $R_i$ moves the value from $\tmpStorage_i$ to $\storage_i$ (Line 32, 33) if the write has been put in $\tmpStorage_i$ before. Otherwise, $R_i$ updates $\storage_i$ to make sure that $j$'s write has the correct tag. The tag in $\storage_i$ could be stale if both $R_i$ and $j$ have not observed a previously completed write operation. Next, $R_i$ proceeds with the steps similar to the previous handler: updates tag (Line 36, 37) and view (Line 38), and notifies others that it has learned a new value (Line 39). Finally, $R_i$ sends acknowledgement to $j$ (Line 40).

\vspace{3pt}
\noindent\textbf{Technical Challenge 1}:~Due to asynchrony and failure, it is possible for a write to have a \textit{stale speculative timestamp}. Consider the example in Figure \ref{fig:speculative-ts}, node $R_3$ has not observed the most recent write $write_1$; hence, its timestamp is still $1$. Then, $write_2$, invoked by a writer co-located with $R_3$, has a stale tag because its speculative timestamp is less than the one included in a \textit{completed} write operation $write_1$. Recall that %\name{} uses tags to order the writes. Therefore, 
to satisfy linearizability, a read that occurs after $write_2$ has to return the value of $write_2$, instead of $write_1$. 

\name{} achieves this by introducing the second phase to identify and update the correct timestamp, which equals to $3$ in this example. After $write_1$ completes, $R_1$ and $R_2$ have timestamp $2$; hence, after the first phase, $write_2$ learns the most recent timestamp from either node, and updates the correct tag in the second phase.

\noindent\textbf{Read operation}:~ 
In \name{}, a reader can retrieve value from its co-located node. The only task is to figure out the value associated with the most recent tag, i.e., the version of the value that satisfies both real-time and total ordering constraints. \name{} achieves this by first contacting a quorum of nodes to learn their most recent tag $tag^{max}$ (Line 16 -- 18), and using  Condition \textsc{SafeToRead}, as per Definition \ref{def:safetoread}, to obtain the return value (Line 19, 20).

\textit{Background Handler for Reads}:~~Upon receiving $\langle \readData \rangle$ message, node $R_i$ returns its tag $tag_i$ (Line 42), which is the largest known tag at $R_i$. Upon receiving $\langle\updateView\rangle$ message, node $R_i$ updates the corresponding entry in $View_i$ (Line 44). By definition of $View_i$, adding $(ts, k)$ to $View_i[j]$ means that node $R_i$ learns that node $R_j$ has added writer $k$'s value associated with $ts$ to $\storage_j$. Owing to the usage of speculative timestamp, it is possible that $View_i[j]$ has both $(ts, k)$ and $(ts', k)$ for the same $k$'s write operation where $ts \neq ts'$. However, this does not affect the correctness, as we explain next how \name{} addresses technical challenge 2.

\begin{figure}[t]
	%\includesvg[width=\linewidth]{img/speculative-ts.svg}
    \includegraphics[width=\linewidth]{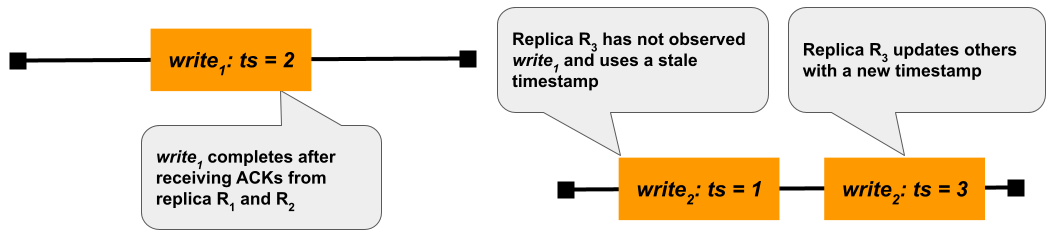}
	\vspace{-20pt}
	\caption{\textit{Speculative Timestamp with} $n=3$. Two ends denote the invocation and the completion time of each operation, respectively. Orange box denotes the timestamp (\textit{ts}) field of each write's tag. $write_2$'s timestamp was stale initially.}%In the first phase, $write_2$'s speculative timestamp is stale.}
	\label{fig:speculative-ts}
    \vspace{-15pt}
\end{figure}

\begin{definition}%[Condition \textsc{SafeToReturn}]
\label{def:safetoread}
Condition \textsc{SafeToRead} is said to hold on $\view_i,~tax^{max}$ and $value$ if there exists a $(tag, value)$ in $\storage_i$ where $tag \geq tag^{max}$, and there is a quorum $Q_R$ s.t. (i) $i \in Q_R$; and (ii) for each $j \in Q_R, (tag, j) \in \view_i[j]$.

% \begin{itemize}[noitemsep,nolistsep]
%     \item $i \in Q_R$; and 
%     \item for each $j \in Q_R, (tag, j) \in \view_i[j]$.
% \end{itemize}
\end{definition}

Intuitively, Condition \textsc{SafeToRead} finds a $(tag, value)$ pair in $\storage_i$ whose $tag$ is larger than $\currentTag^{max}$ and $value$ is received by a quorum of nodes $Q_R$, including $R_i$. In other words, the condition ensures that the returned value is received by a read quorum $Q_R$, and its version $tag$ is at least as recent as $\currentTag^{max}$. 
%The reason that $R_i$ has to be in the read quorum $Q_R$ is that $(tag, value)$ can be found in $\storage_i$. 

Figure \ref{fig:read} presents the fast and slow paths for reads. If $tag^{max} = tag_1$, then the condition must hold at that moment. If there is a concurrent write (with a larger tag), then $tag^{max} > tag_1$. Thus, the reader at $R_1$ needs to wait for more messages -- $\langle\writeData\rangle$ from $R_5$ and $\langle\updateView\rangle$ from two other nodes --  to satisfy Condition \textsc{SafeToRead}. In the worst case, this takes 2 RTT.%, as shown in the example. 

\noindent\textbf{Technical Challenge 2}:~With speculative timestamp, a write may have two tags (or timestamps). We say that a write is ``\textit{associated}'' with a tag $tag$ (or timestamp $ts$) if a read returns the value of a write with $tag$ (or $ts$). In the example of Figure \ref{fig:speculative-ts}, we do not want $write_2$ to be associated with timestamp 1, i.e., no read should return the value of $write_2$ with timestamp $ 1$. This is because that eventually $write_2$ will update its timestamp to $3$, which means $write_2$ will be associated with two tags. Consequently, it is impossible to find a total ordering using associated tags that satisfies linearizability. 

Condition \textsc{SafeToRead} is devised so that such an undesirable scenario can never occur. %with $n=3$ or $5$. 
In \name{}, a read returns $value$ if a read quorum $Q_R$ has received $(tag, value)$. In the aforementioned example, no read can return a value associated with timestamp 1 because $R_1$ and $R_2$ observe $ts=1$ being stale, and $R_3$ updates $ts$ \textit{only after} $write_2$ is completed; hence, it is not possible to gather a read quorum. When $n=3$, if a write observes a stale tag $tag^{old}$, then no read can return a value with $tag^{old}$. This is because at most one other node would consider $tag^{old}$ as the most recent tag, which means that no read can obtain $tag^{old}$ from a read quorum $Q_R$. 

\subsection{Correctness and Performance Analysis}

We follow the proof structure in \cite{GeoQuorum_Lynch1997,AttiyaBD1995}, i.e., using tags to assign the order of the operations. The key difference is to prove that \name{} addresses Technical Challenge 2 correctly -- each write can only be associated with one tag. We prove the claim by formalizing the argument in Section \ref{s:protocol}. The complete proof is presented in Appendix \ref{app:proof}. 

\name{} achieves optimistically fast operations, i.e., both writes and reads take 1 RTT if there is no concurrent write.  Both operations take 2 RTT in the worst case, as shown in Figure~\ref{fig:read}. Message complexity for reads is the same as prior algorithms \cite{GeoQuorum_Lynch1997,AttiyaBD1995}, $O(n)$. For reads, we only count the messages on the fast path, since as shown in Figure \ref{fig:read}, other messages for committing reads belong to writes. For writes, the message complexity is $O(n^2)$ due to $\langle \updateView \rangle$. 
Despite higher complexity, we find this acceptable in our target case because this design allows for using the fast path for reads. Moreover, for the case of object storage systems, $\langle \updateView \rangle$ only contains tag, not the data itself. Since typical data size is in the range of KBs, MBs or even more  \cite{Giza_ATC2017,IBM-Cloud-Object_PDSW2015,IBM-Cloud-Storage_HotStorage2020}, the bit complexity and network bandwidth consumption of the overhead are negligible.

\subsection{The Case of $n=4$ or $5$}
\label{s:n=5}

Algorithm \ref{algo:Gus} does \textit{not} work with $n>3$ owing to \textit{Technical Challenge 2} -- a write could be associated with two tags when $n>3$. Consider the example in Figure \ref{fig:write-read-example}. Suppose $write_1$ is from a writer $W_1$ at $R_1$ and $write_2$ is from writer $W_2$ at $R_2$. Writer $W_1$ learns from $R_2$ that its speculative timestamp is stale due to the concurrent $write_2$. In the meantime, $R_3, R_4$, and $R_5$ have not observed $write_2$ and form a read quorum which allows a reader to read $write_1$ with a stale timestamp. After $W_1$ updates a new timestamp due to the notification from $R_2$, $write_1$ is associated with two tags.

To address this issue, more information needs to be included in $\langle \ackWrite \rangle$ message -- if the highest tag is from $R_j$, then $R_j$ needs to indicate whether a write is completed or not. 
% Concretely, the writer enters the second phase if any of the following conditions on the responses from the other two nodes (in the quorum $Q$ at Line 4) is satisfied:
% \begin{itemize}[noitemsep,nolistsep]
%     \item Both nodes find the tag stale.
%     \item At least one node responds that the highest tag is due to its writer, and the write has completed. 
%     \item At least one node responds that the highest tag is due to a writer at another node.
% \end{itemize}
% \noindent Otherwise, writer does not update the tag. 
In the earlier example, the second phase is not needed. Since $write_2$ has not completed yet (i.e., $W_2$ has not received a confirmation from a quorum), the writer $W_1$ does \textit{not} need to update the tag, and can complete its write on the fast path. This does not violate linearizability, since by definition, two concurrent writes can appear in any order. %in this example. 

% Intuitively, these three rules ensure that when there is a complete write before the writer $w_i$ invokes its operation, $w_i$ will update its tag to satisfy the real-time constraint. Moreover, it also address technical challenge 2 by ensuring that whenever $w_i$ updates the tag, no read can associate the stale tag with its write. This is because (i) by design, at most two nodes other than node $R_i$ think $w_i$'s tag is not stale; and (ii) $R_i$ will not update its tag until $w_i$ completes its write. Consequently, no reader can gather a read quorum with a stale tag in this case. 

%\subsection{Impossibility of a Larger $n$}

%\subsection{The Case of $n=7$ and Beyond}
\section{Impossibility}
\label{s:impossible}

\begin{figure}[t]
	%\includesvg[width=\linewidth]{img/read-eg.svg}
    \includegraphics[width=\linewidth]{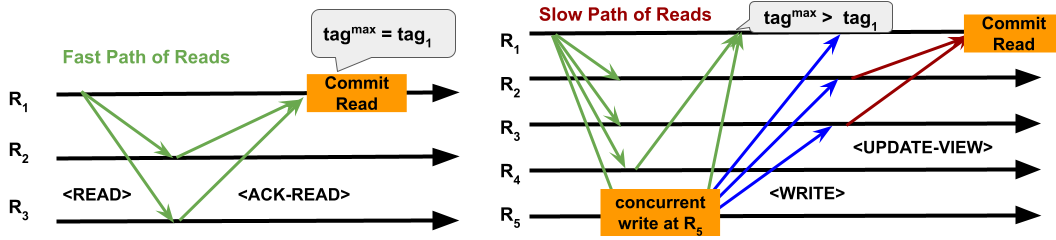}
	\vspace{-20pt}
	\caption{Fast/Slow Path of Reads. Green arrows represent $\langle\readData\rangle$ and $\langle\ackRead\rangle$, blue represent $\langle\writeData\rangle$, and red represent $\langle\updateView\rangle$. (On the right figure, not all messages are shown for brevity.) }%\\If $tag^{max} > tag_1$, then reader at $R_1$ needs to wait until that a quorum $Q_R$ has received the concurrent write at $R_5$ and the associated $tag$, which is $tag^{max}$ in this example.}
	\label{fig:read}
\end{figure}

\begin{figure}[t]
	%\includesvg[width=\linewidth]{img/write-read.svg}
    \includegraphics[width=\linewidth]{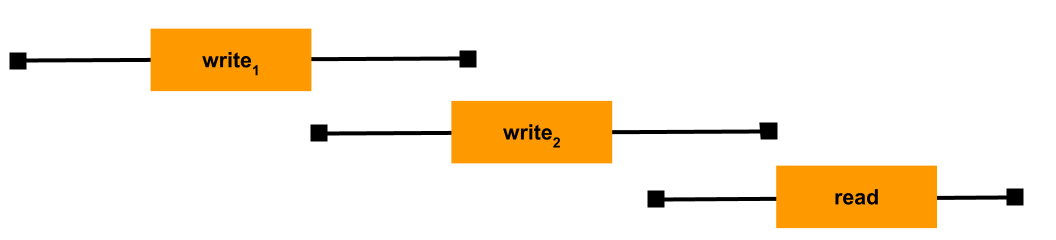}
	\vspace{-25pt}
	\caption{Example Execution. %Two writes are concurrent with each other, and $write_2$ and $read$ are concurrent. 
	 $write_2$ is concurrent with the other two operations.}
	\label{fig:write-read-example}
\end{figure}

%This section proves the following theorem.

\begin{theorem}
For $n > 5$ and $n = 2f+1$, it's impossible to have an atomic register that supports optimistically fast writes and reads. 
\end{theorem}

\begin{proof}[Proof Sketch]
The proof is based on an indistinguishability argument, which constructs several executions indistinguishable to nodes such that in one of the executions, a reader has to return a value that violates linearizability. All the executions we construct have no concurrent write; hence, the optimistically fast operations require all operations to complete in 1 RTT. 

Consider $n=7$ with nodes $R_a$ to $R_g$. Since $f=3$, the maximum quorum size to ensure liveness is $4$. Now, consider the following executions such that the first write $w_1$ is invoked by a writer at node $R_e$ and writes value $x$:

\begin{itemize}[noitemsep,nolistsep]
    \item E1: $w_1$ is completed with a write quorum $\{R_d, R_e, R_f, R_g\}$. All the messages from the write quorum to other nodes are delayed, except for the messages between $R_d$ and $R_a$. At some time $t$ after $w_1$ completes, reader at node $R_a$ invokes a read $r_1$ that completes with a read quorum $\{R_a, R_b, R_c, R_d\}$ and returns $w_1$'s value $x$.

    \item E2: Only node $R_d$ receives $w_1$, because node $R_e$ and its writer client crash during the write. The messages from $R_e$ are all lost, because it has crashed. The messages from $R_d$ to nodes other than node $R_a$ are delayed. 
    At time $t$, reader at node $R_a$ invokes a read $r_1$ that completes with a read quorum $\{R_a, R_b, R_c, R_d\}$. Since E1 and E2 are indistinguishable from the perspective of node $R_a$ and its reader, the read returns $x$.

    \item E3: Now, we construct E3 by extending E2. Right after the read $r_1$ completes, nodes $R_a$ and $R_d$ crash. This is allowed since $f=3$. Furthermore, the messages from $R_a$ and $R_d$ to all the other nodes are lost, because they have crashed. As a result, none of $R_b, R_c, R_f, R_g$ learns the existence of $w_1$.
    
    At some later time, a reader at $R_b$ invokes another read $r_2$ that completes with a read quorum $\{R_b, R_c, R_f, R_g\}$, and returns a default value, violating linearizability. 
\end{itemize}
It's straightforward to extend the argument to a larger $n$.
\end{proof}

Fundamentally, the impossibility is because that the quorum intersection is too small for a larger $f$. Due to the 1-RTT communication, readers or writers are \textit{not} able to update all nodes in the read or write quorums. This is why we can defer messages in the proof. In the case of ABD, the read quorum of $r_1$ will learn the most recent value before $r_1$ completes because of the write-back. Consequently, the read quorum for $r_2$ would return $x$. 

Note that in the construction above, we require 3 nodes to fail. This is why \name{} works for $n \leq 5$. For example, when $n=3$, the union of the reader and the quorum intersection is enough to ensure safety; hence, E3 is impossible and readers can learn the correct value that satisfies linearizability. To circumvent the impossibility, we either need to sacrifice optimistically fast operation or resilience. %we present two solution in the next section. To increase scalability, the first one sacrifices optimistically fast operation, whereas the second sacrifices optimal resilience. 

% \noindent\textbf{Impossibility:}~
% In general, it is impossible to achieve all three properties when $n=2f+1$ and $n > 5$. Consider $n=7$. Suppose a writer $W_i$ is notified with a stale tag by one node $R_j$ in $Q$, which is due to another writer at node, say $R_k$, that is \textit{not} in $Q$. Now $W_i$ is facing a dilemma:

% \begin{itemize}[noitemsep,nolistsep]
%     \item \textit{Updating the tag}: it is possible that nodes other than $R_i, R_j$ and $R_k$ already served a read with a stale tag. In this case, the write is associated with two tags.
    
%     \item \textit{Not updating the tag}: it is possible that the write from $R_k$ has indeed completed, and $W_i$'s write violates the real-time constraint. 
% \end{itemize}
% Since $W_i$ cannot distinguish from the two scenarios, no matter the action $W_i$ takes, safety is violated. 

	\section{Practical Considerations}
\label{s:optimization}

%We present several practical optimizations in this section. 

\subsection{Scalability}
%\subsection{Increasing Scalability in \name{}}
\label{s:optimization-scalability}

To increase scalability, we present two solutions for $n > 5$. The first increases 1 RTT for writes, which is suitable for serving larger objects because of its natural integration with erasure coding. The second increases the quorum size by focusing on the case of a smaller number of concurrent failures (relaxed resilience), a common case for modern geo-replicated systems  \cite{Tempo_Alexey_Eurosys21,Atlas_Sutra_Eurosys20,Spanner_OSDI12}. For $n \leq 5$, these solutions are not needed, and therefore not applied.

\noindent\textbf{Layered design by separating metadata and data}:~
Inspired by Giza \cite{Giza_ATC2017} and Layered
Data Replication \cite{LDR_Lynch_DISC2003}, we integrate a layered design with \name{}, which separate the data and metadata paths into two layers. To read, a client first contacts the servers in the metadata layer to find the set of data servers that have the most recent data, and then reads the data from any of them. To write, a client first writes its data to a set of data servers, then update the metadata servers. Such a layered design allows the underlying data/metadata servers to optimize different workloads and features. Giza uses Azure object storage as the data server and Azure table as the core of metadata layer. 
%and use erasure coding to store large objects. The system consists of two layers -- meta-data (the version, the IDs and the locations of the object) and data. 
%uses replication to update and fetch the meta-data before updating and fetching the coded fragments. 
Giza adopts Fast Paxos \cite{lamport2006fast} to replicate the metadata (i.e., the version, the IDs and the locations of the object) to 3 or 5 metadata servers, whereas \name{} uses Algorithm \ref{algo:Gus}. In our implementation, we use Redis as the data server because of its high performance and support of durability. 

As observed in \cite{Giza_ATC2017,Pando_NSDI2020,LEGOStore_Viveck2021}, to save storage and network cost, it is common to use erasure coding for the data layer. 
%To increase redundancy, we adopt and implement the approach proposed by Giza \cite{Giza_ATC2017}, which uses erasure coding and separate the data and meta-data paths. 
For larger objects, we adopt the $n = k + m$ Reed-Solomon code \cite{EC_book_Cambridge2003} -- the value is divided into $k$ data fragments, and the encoder generates $m$ parity fragments. Each data server  stores exactly one fragment. The object is durable as long as at most $m$ node fails. With erasure coding, both writes and reads take 2 RTT on the fast path.

\vspace{3pt}
\noindent\textbf{Increasing quorum sizes by lowering resilience}:~Concurrent failures in replication across datacenters are rare and transient \cite{Atlas_Sutra_Eurosys20,Tempo_Alexey_Eurosys21,Spanner_OSDI12}; hence, it is reasonable to focus on a smaller $f$ with a larger $n$. 
Let $Q_R$ and $Q_W$ be the size of the read and write quorum, respectively. As long as they satisfy the following inequalities, \name{} ensures safety: \hspace{24pt} $2Q_W > n$~~~\text{and}~~~$2n - 2Q_W -1 < Q_R$ %Liveness is relaxed in the sense that by increasing the quorum size, \name{} decreases the ability to tolerate \textit{simultaneous} failures -- a read (write) operation can be completed as long as $n-Q_R$ ($n-Q_W$) replicas are alive; otherwise, an operation needs to wait for the reconfiguration mechanism to replace failed replica(s). Fortunately, as observed in \cite{Atlas_Sutra_Eurosys20,Tempo_Alexey_Eurosys21}, simultaneous failures is very rare in practical geo-replicated systems. 

% \begin{equation*}
% \label{eq:write-read-intersect}
%\label{eq:write-intersect}
%\hspace{41pt} $2Q_W > n$~~~\text{and}~~~$2n - 2Q_W -1 < Q_R$

% \end{equation*}

% \begin{equation}
% \label{eq:write-read-intersect}
%     2n - 2Q_W -1 < Q_R
% \end{equation}

\noindent The first part ensures that any two write quorums intersect with each other, whereas the second part %Equation (\ref{eq:write-read-intersect})
ensures that any write can only be associated with one tag (which can be argued similar as before). %\name{} tolerates $n-\max\{Q_W, Q_R\}$ concurrent failures.
As long as a writer (or reader) can reach a write (or read) quorum, then its operation can be completed. %If more data centers fail, then operations may be delayed until failed replica(s) are replaced using the reconfiguration protocol (Section \ref{s:optimization-reconfig}).

For read-intensive workloads, we can let $Q_R$ be %the simple majority, namely 
$\lfloor n/2 \rfloor+1$. Then  $Q_W > n-\frac{\lfloor n/2\rfloor}{2} - 1$. For write-intensive workloads, we can lower write quorum size by increasing read quorum size accordingly. In other words, tolerating less failures allows \name{} to explore a trade-off between quorum sizes and performance of different operations.

\begin{table*}[t]
\centering
\begin{tabular}{l|c|c|c|c|l}
 & \multicolumn{2}{c|}{Latency} & Fast-Quorum Size  & Optimistically &  Limitation\\ \cline{2-3}
& \multicolumn{1}{c|}{Read}        & Write       &      &  Fast Ops     &   \\ \hline
EPaxos \cite{EPaxos_SOSP13} & \multicolumn{1}{c|}{1/2}          & 1/2           & $f+\lfloor (f+1)/2  \rfloor$      & read/write     & tail latency, dependency tracking \\\hline
%\textsc{Tempo} \cite{Tempo_Alexey_Eurosys21} & \multicolumn{1}{c|}{1/2}          & 1/2           & $f+\lfloor n/2  \rfloor$      & -     & & optimal commit latency when $f=1$ \\\hline
Gryff \cite{Gryff_Lloyd_NSDI2020}   & \multicolumn{1}{c|}{1/2}          & 2/2           & $f+\lfloor (f+1)/2  \rfloor$   & read       & write latency, throughput \\\hline
Giza \cite{Giza_ATC2017}   & \multicolumn{1}{c|}{1/2}          & 1/2           & $\lfloor 3n/4  \rfloor + 1$   & -     &  coordinator, large fast-quorum \\\hline
Gus (this work)    & \multicolumn{1}{c|}{1/2}          & 1/2           & $\lfloor n/2  \rfloor + 1$       & read/write    & all properties when $n=3$ or $5$%\\\hline
%Gus (this work)    & \multicolumn{1}{c|}{1/2}          & 1/2           & $\lfloor n/2  \rfloor + 1$       & read  &  \checkmark   & for $n \geq 5$, a smaller $f$, write latency
\end{tabular}
\caption{\name{} vs. related leaderless systems designed for geo-redundancy. The first number in the Latency column indicates the RTT in the common case (fast path), and the second is the RTT with contention. All of the systems tolerate $f$ crashes with $n=2f+1$ nodes.  Both EPaxos and Gryff support read-modify-write, whereas Giza and \name{} do not. When $n$ is beyond 5, some properties of \name{} no longer hold.}
\label{t:comparison}
\end{table*}

\subsection{Optimizing Reads in \name{}}
We have two approaches to optimize reads in \name{}. Consider the case of $n=3$. \name{}'s read only needs 1 RTT with one simple change -- piggybacking the value associated with the highest tag in $\langle \ackRead \rangle$ at Line 42. Since any two nodes form a read quorum $Q_R$, upon receiving the value associated with $tag^{max}$, the reader can update $View$ and directly return the value, which must satisfy Condition \textsc{SafeToReturn}. %This optimization reduces both median and tail latency. 
The second optimization can be applied to the case when $n \leq 5$ and when a node serves several reader clients (a typical case in practical systems). Observe that read does not change the state at other nodes; hence, when there are multiple concurrent readers co-located in the same data center, then all the subsequent reads can ``tag along'' the first read \textit{without} sending any messages. %This optimization reduces mainly median latency in our targeted workload, since at the tail-end, reads are not concurrent.  

% \subsection{Reconfiguration}
% \label{s:optimization-reconfig}

% For practical systems, it is important to support reconfiguration (or dynamic group membership). For example, administrators need to remove and replace a failed replica to ensure high availability.  \name{} follows the architecture of ABD \cite{AttiyaBD1995}; thus, prior reconfiguration mechanisms designed for ABD (e.g.,  \cite{ReconfigStorage_tutorial2010,Rambo_Lynch_DISC2002,RamboII_Seth_DSN2003}) are compatible with \name{}. Intuitively, replicas use consensus to agree on a sequence of configurations that specify the set of participating replicas. Administrators can create a new configuration by adding or removing replicas. Each replica also stores ``forward pointer'' to new configurations so that clients with stale information can catch up to the most recent configuration and start accessing \name. 

\section{Evaluation}
\label{s:evaluation}

We evaluate \name{} in practical settings. Our evaluation is focused on the case of geo-replicated object storages, because (i) atomic registers capture its semantic \cite{Giza_ATC2017,Gryff_Lloyd_NSDI2020}; and (ii) round-trip time matters the most for user-perceived latency in the case of geo-replication, as cross-datacenter latency can be in the order of 100+ms. 

As discussed in Section \ref{s:related}, prior algorithms \cite{ABD_Analysis_PODC2020,ABD_read_Rachid_PODC2004,ABD_write_OPODIS2009} with fast operations have limited practical usages due to their stringent conditions. Therefore, we compare \name{} with consensus-based systems. Even though these systems only ensure liveness when the network is partially synchronous, they have high-performance in common cases. We first present related systems that are optimized for geo-replication, followed by our evaluation.

\subsection{Related Work: Geo-replicated System}
\label{s:related2}

A comparison of \name's features against state-of-the-art competitors is outlined in Table~\ref{t:comparison}. To ensure a total ordering, storage systems often adopt the consensus-based approach. Most production systems \cite{Megastore_CIDR2011,Azure_SOSPI11,CockroachDB,Spanner_OSDI12,etcd,Physali_NSDI2020} rely on variants of Paxos \cite{lamport1998part,lamport2001paxos} or Raft \cite{Raft_ATC14} for agreeing on the order of client commands (or requests) and execute the commands following the agreed order. Unfortunately, these leader-based consensus protocols suffer long latency -- 2 RTT (cross-datacenter message delay) -- if the clients are not co-located with the leader data center.  

Many recent systems \cite{Atlas_Sutra_Eurosys20,M2Paxos_DSN16,EPaxos_SOSP13,Caesar_DSN17,Mencius_Marzullo_OSDI08,Tempo_Alexey_Eurosys21} propose a leaderless design to avoid the bottleneck at the leader and achieve optimistically fast operations.\footnote{It is also called ``optimal commit latency'' in \cite{EPaxos_SOSP13}; however, the term is typically used for consensus-based systems. Hence, we use a different term to avoid confusion.} 
EPaxos commits commands in 1 RTT when there is no contention, and 2 RTT with contention. 
Unfortunately, EPaxos has worse tail latency than Paxos-based systems (up to 4x worse) \cite{Gryff_Lloyd_NSDI2020} and may have a livelock in pathological cases \cite{EPaxos_revisited_NSDI2021}. This is mainly because EPaxos's fine-grained dependency tracking may chain dependency recursively, and the execution of some operations may be delayed in wide-area networks  \cite{EPaxos_revisited_NSDI2021}. 

Gryff \cite{Gryff_Lloyd_NSDI2020} reduces tail latency by unifying consensus and shared registers. Gryff implements an abstraction that provides read, write and read-modify-write (RMW) on a single object. On a high-level, it uses ABD register \cite{AttiyaBD1995} to process reads and writes, and EPaxos \cite{EPaxos_SOSP13} to process RMWs. 
While Gryff reduces p99 read latency compared to EPaxos, it always takes 2 RTT to complete a write; hence, it does not achieve optimistically fast writes and is not suitable for write-intensive workloads like game hosting or enterprise backup service that typically has around 90\% of writes \cite{IBM-Cloud-Object_PDSW2015}. 

Giza uses Fast Paxos \cite{lamport2006fast} to agree on the version for each operation, and needs only 1 RTT when there is no concurrent write. Two downsides of Giza are its reliance on the coordinator to order concurrent write operations and that its fast-quorum requires a super majority. Both affect tail latency, especially for the geo-replicated storage systems, because the clients need to wait for the nodes or the coordinator in the further datacenters. 

\textsc{Atlas} \cite{Atlas_Sutra_Eurosys20} and \textsc{Tempo} \cite{Tempo_Alexey_Eurosys21} are two recent consensus-based systems that sacrifice resilience to optimize performance. \textsc{Atlas} uses dependency tracking; hence, suffers from long tail latency. \textsc{Tempo} develops a novel mechanism of using (logical) timestamps to determine when it is safe to execute a particular operation. 
Both systems have quorum size $\lfloor n/2 \rfloor + f$, which is optimal when $f=1$. \textsc{Atlas} and \textsc{Tempo} do not distinguish between read and write quorums. Compared to them,  \name{} can be configured to have an optimal read quorum size, while having the write quorum size the same or greater by 1. Table \ref{t:quorum-size} presents some examples. \name{}'s smaller read quorum not only allows a better read latency, but also ensures that reads can still complete, when $\geq \lfloor n/2 \rfloor+1$ nodes are alive. For the case of $n=11$, \textsc{Tempo} requires a quorum of $8$, which equals to the write quorum of \name. Reads can be served with a quorum of $6$ in \name.  Later in Section \ref{s:scalability}, we will see how a smaller quorum size allows \name{} to have better tail latency under practical workloads. 
%The larger the $n$ and $f$ are, the larger the difference between the read quorum size in \name{} and the quorum size in \textsc{Atlas} and \textsc{Tempo} is. 

\begin{table}[t]
\centering
\begin{tabular}{ll||l|l|c}
   &    & \multicolumn{2}{l|}{~~~\name} & \textsc{Atlas}/\textsc{Tempo} \\\hline
$n$  & $f$ & $Q_R$         & $Q_W$         & $Q$     \\\hline\hline
7  & 2  &   4      &   5         &    5   \\
9  & 2  &   5     &     7       &     6  \\
11 & 3  &   6      &     8      &   8 \\
13 & 3  &   7      &     10      &   9 \\
\end{tabular}
\caption{Read/write quorum size ($Q_R/Q_R$) in \name{}, and quorum size ($Q$) in \textsc{Atlas} and \textsc{Tempo}, where $f=$ number of tolerated concurrent failures.}
\label{t:quorum-size}
\vspace{-30pt}
\end{table}

%\vspace{3pt}
%\noindent\textbf{Leaderless SMR Systems.}~Many recent leaderless SMR-based systems have achieved the optimal commit latency 
Other consensus-based systems achieve optimistically fast operations for both reads and writes, e.g., M$^2$Paxos \cite{M2Paxos_DSN16}, Caesar \cite{Caesar_DSN17}, and Mencius \cite{Mencius_Marzullo_OSDI08}. Each system performs well in certain cases. To support more general operations, e.g., transactions or RMW, they sacrifice high-performance under high skewed workload. Both EPaxos and Caesar use dependency tracking, which leads to high tail latency \cite{EPaxos_revisited_NSDI2021}.  M$^2$Paxos requires a lock on an object; hence, not suitable for workloads with high contention. Mencius need information from all nodes. %\textsc{Pando} \cite{Pando_NSDI2020}  enables near-optimal latency versus cost tradeoff using erasure coding; however, the near-optimal latency is only ensured when there is no conflicting write. It still requires the use of a leader to resolve contention. 

\subsection{Implementation and Experiment Setup}

In our evaluation, we focus on tail latency, because it is well-known that user-perceived latency is correlated with the tail latency of the underlying storage systems \cite{Tail_at_scale_CACM2013,Facebook_Memcached_workload_sigmetrics12,Facebook_Memcache_NSDI13,Latency_blog2009}. 
We evaluate \name{} against  two categories of competitors: (i) those aiming/optimizing for fault-tolerant non-blocking MWMR registers (Gryff and Giza), and (ii) state-of-the-art consensus systems (EPaxos and \textsc{Tempo}) that are optimized for the scenarios that Gus is optimized for.

For Gryff, we are essentially evaluating its ABD component (and Gryff's optimizations), as the workload consists of only reads and writes. 
For Giza, we only focus on the tail latency \textit{without} any concurrent write. As documented in \cite{Giza_ATC2017}, its design is not optimized for concurrency. For scalability, we compare \name{} with \textsc{Tempo} so that they tolerate the same number of concurrent failures.

Recent systems \cite{Pando_NSDI2020,CRaft_Fast2020,Nil-Externality_SOSP21} use techniques such as coding and nil-externality to further improve performance. We do not compare against them, due to their leader-based design. We mainly focus on leaderless systems, because as demonstrated in  \cite{EPaxos_SOSP13,Gryff_Lloyd_NSDI2020,Giza_ATC2017}, leaderless systems have better performance in both common case and tail latency in the context of geo-replicated storages.

\vspace{3pt}
\noindent\textbf{Implementation.}~We implemented \name{}\footnote{\url{github.com/bc-computing/gus-automation}} and our version of Giza (source code not available) in Go using the framework of EPaxos and Gryff to ensure a fair comparison between protocols. For \textsc{Tempo}, we use the implementation in \cite{Tempo_Alexey_Eurosys21}. 

Clearly, even though in Algorithm~\ref{algo:Gus} we focus on a single register (or object) for clarity, our implementation supports multiple objects and adopts the optimizations mentioned in Section \ref{s:optimization}. In order to do that, we include two extra fields in each message type -- key $key$ and sequence number $seq$. The key denotes the identifier of each object, and the sequence number is the operation index. This allows \name{} to support multiple objects and also pipelining. We do \textit{not} enable thrift optimization nor batching, because these optimizations generally increase the tail latency, by increasing the chance of conflicts \cite{EPaxos_revisited_NSDI2021,EPaxos_SOSP13,Gryff_Lloyd_NSDI2020}.

In addition, we follow the same setup in \cite{EPaxos_revisited_NSDI2021,EPaxos_SOSP13,Gryff_Lloyd_NSDI2020} to separate node and client machines for best performance. Each node has several client proxies that handle requests from the respective client. 

\begin{table}[]
\centering
\begin{tabular}{l|lllll}
   & CA  & VA  & IR  & OR  & JP  \\\hline
CA & 0.2 &     &     &     &     \\
VA & 72  & 0.2 &     &     &     \\
IR & 151 & 88  & 0.2 &     &     \\
OR & 59  & 93  & 145 & 0.2 &     \\
JP & 113 & 162 & 220 & 121 & 0.2
\end{tabular}
\caption{RTT (in ms) between VMs in
emulated geographic regions \cite{Gryff_Lloyd_NSDI2020}. For $n=3$, we use VMs in CA, VA, and IR.}
\label{t:RTT}
\end{table}

\begin{figure*}[t]
\begin{subfigure}{0.33\textwidth}
  \centering
  \includegraphics[width=\textwidth]{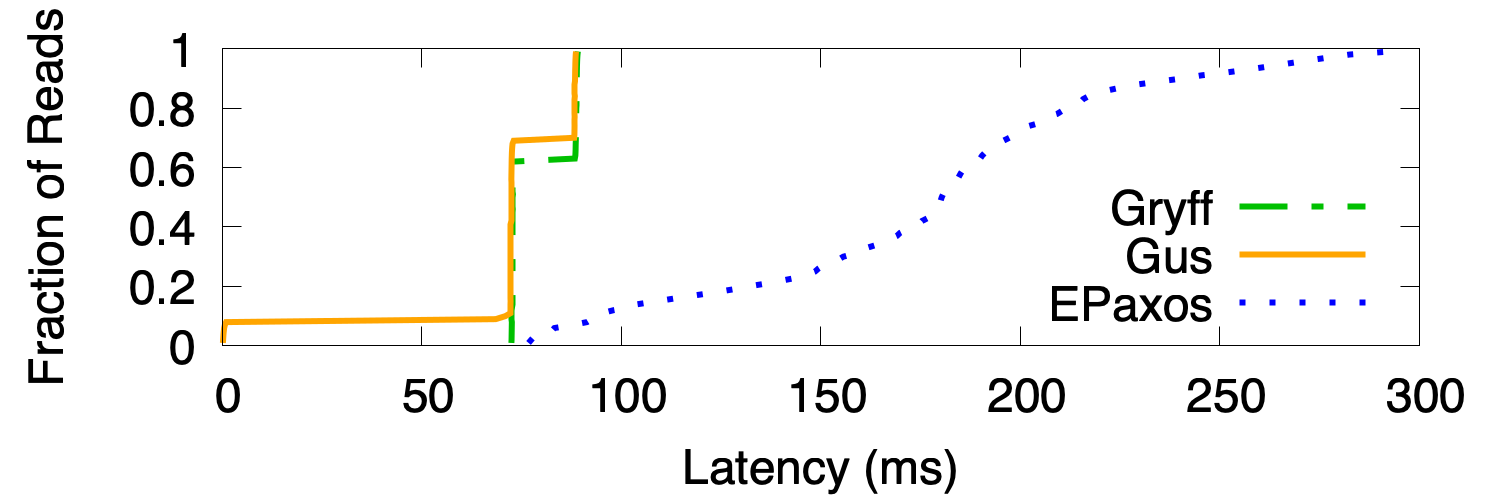}
  \includegraphics[width=\textwidth]{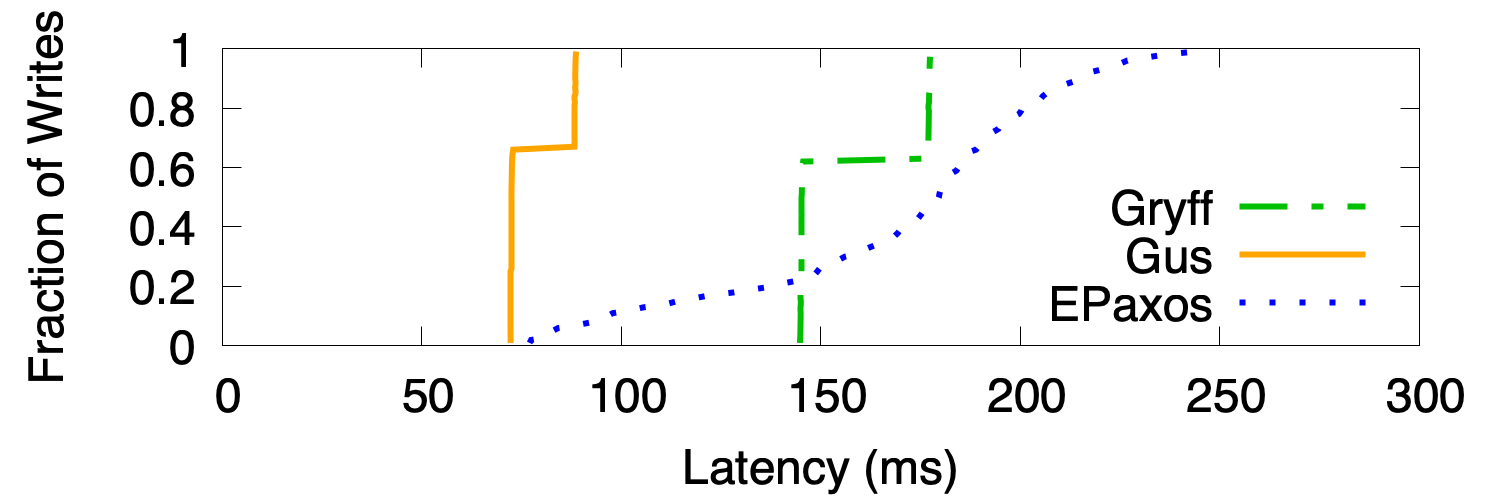}
  \caption{\textbf{2\% Conflict}}
  \label{fig-sub:6a}
\end{subfigure}%
\begin{subfigure}{0.33\textwidth}
  \centering
  \includegraphics[width=\textwidth]{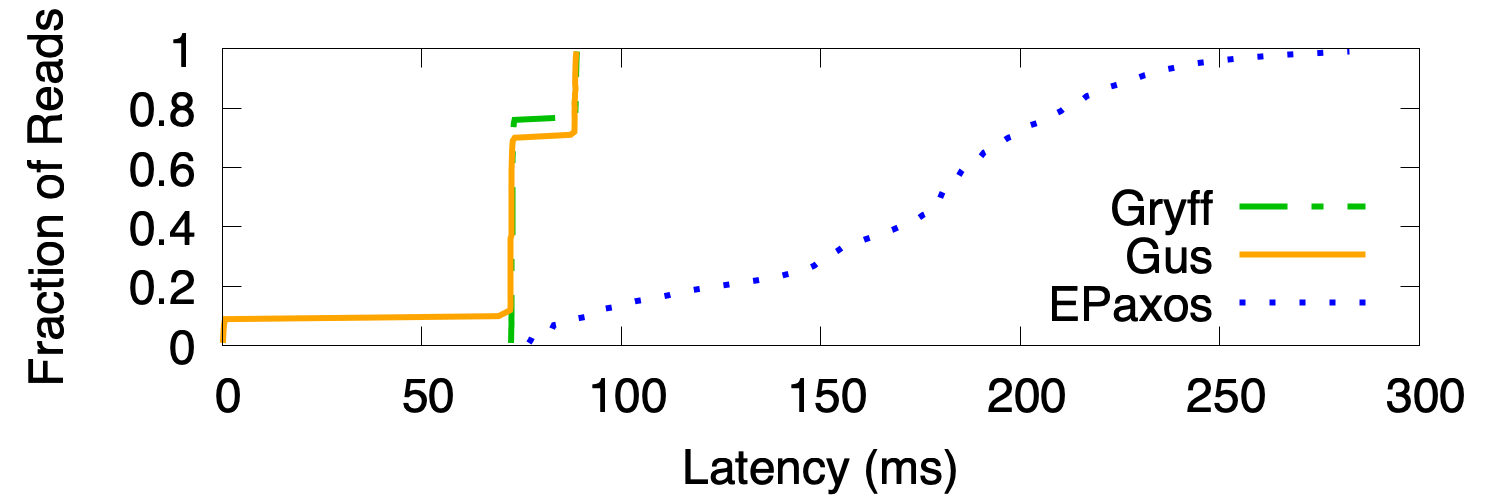}
  \includegraphics[width=\textwidth]{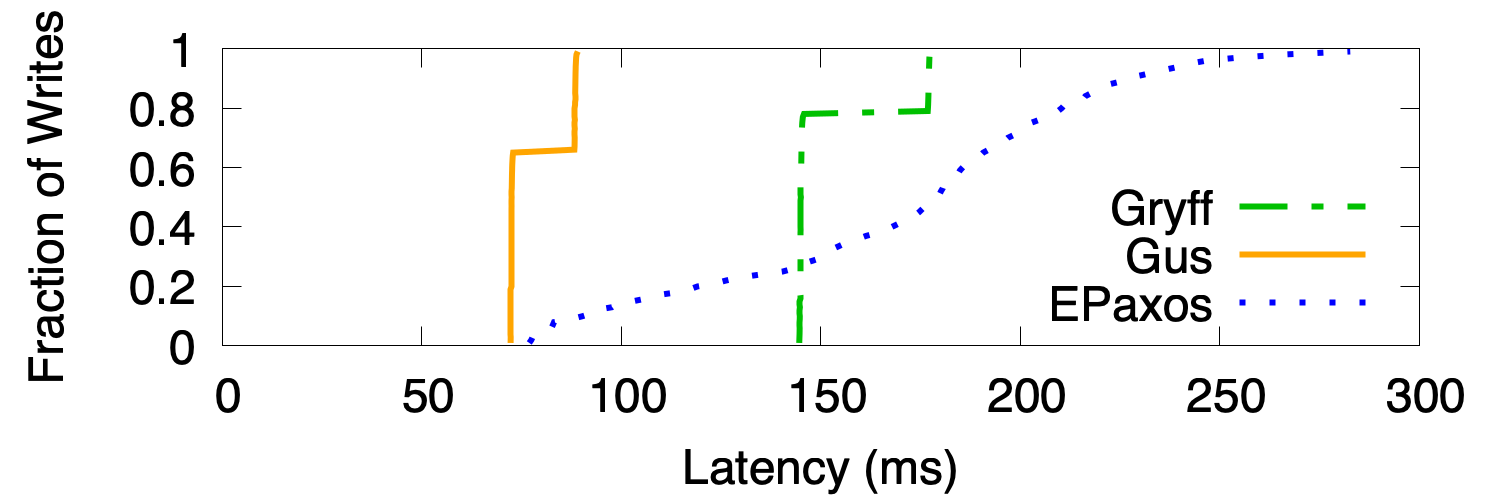}
  \caption{\textbf{10\% Conflict}}
  \label{fig-sub:6b}
\end{subfigure}
\begin{subfigure}{0.33\textwidth}
  \centering
  %\includegraphics[width=\textwidth]{img/Rabia-eval-tho-vs-lat-p50-100.pdf}
  %\caption{Throughput vs. Median Latency.  Both systems have client batch size $=100$.}
  \includegraphics[width=\textwidth]{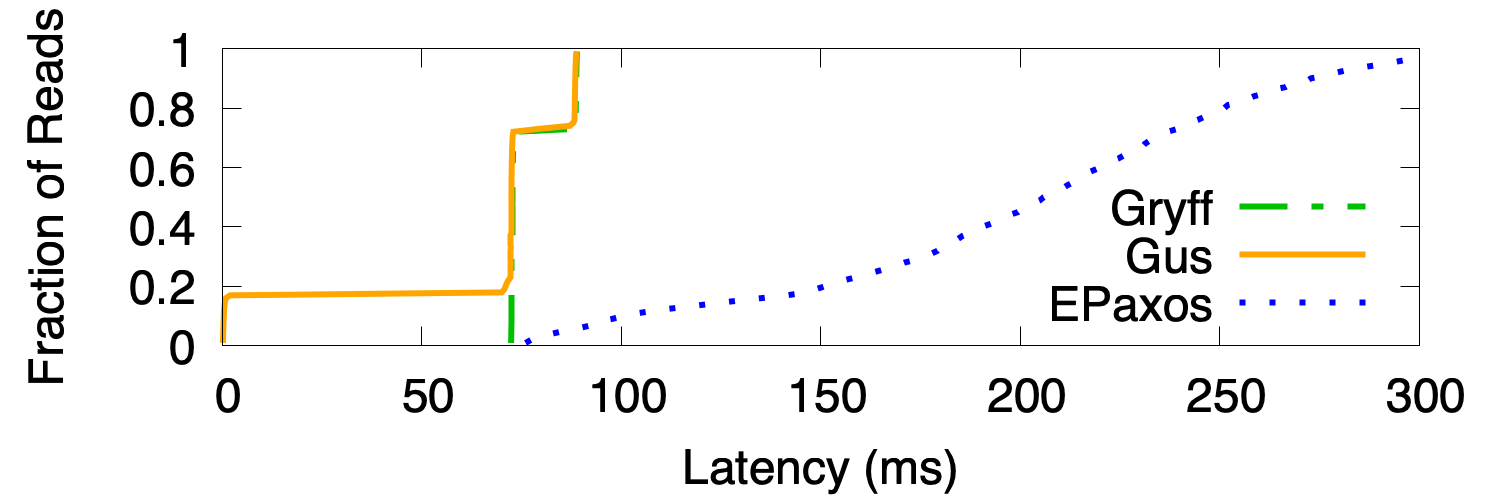}
  \includegraphics[width=\textwidth]{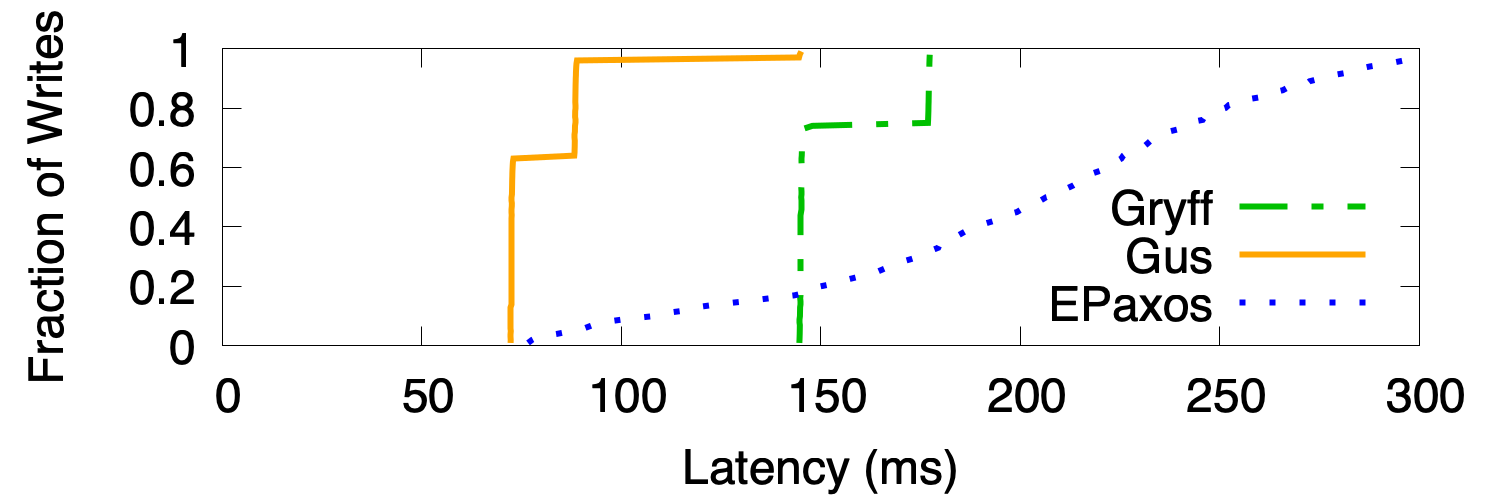}
  \caption{\textbf{25\% Conflict}}
  \label{fig-sub:6c}
\end{subfigure}

%
%\caption{Thpt-latency, the number in the parentheses means client batch size, proxy batch = 20}
\vspace{-10pt}
\caption{Latency CDF ($n=3$, 94.5\% reads, 5.5\% writes). Both Gryff and \name{} complete reads in 1 RTT when $n=3$.}
\label{fig:6}
\end{figure*}

\vspace{3pt}
\noindent\textbf{Testbed.}~ We run our experiments on CloudLab \cite{CloudLab_ATC19} using m510 (Intel Xeon D-1548, 8 cores, 6GB RAM) for node VMs and c6525-25g (AMD EPYC 7302P, 16 cores, 8GB RAM) for the client VM. We adopt the same latency profile used in \cite{Gryff_Lloyd_NSDI2020} -- (i) $n=3$: nodes in California (CA), Virginia (VA), and Ireland (IR); and (ii) $n=5$: two more
nodes in Oregon (OR) and Japan (JP). The latencies of the wide-area network are emulated using Linux’s Traffic Control (tc) by adding delays to packets on all nodes with filters on different IPs. Table \ref{t:RTT} shows the configured RTT between nodes in different regions. These numbers were chosen to represent typical RTT between the corresponding Amazon EC2 availability regions \cite{Gryff_Lloyd_NSDI2020}. 

\vspace{3pt}
\noindent\textbf{Experiment Setup.}~In all experiments except for the ones in Section \ref{s:erasure}, the clients run on one client VM, which has no artificial latency to all the other node VMs. In experiments testing the integration of the layered design (Section \ref{s:erasure}), the clients are located in the CA region. For most experiments, we use 16 closed-loop clients co-located with each node, again following \cite{Gryff_Lloyd_NSDI2020}. This setup balances between capturing the effect of concurrent operations and avoiding saturating the system. This also allows us to isolate limitations of the hardware and software. We use different numbers of clients to stress the systems in the throughput experiment.

In our implementation, despite the fact that clients and servers are indeed co-located, clients do not interact directly with the system's backend but pass through a proxy interface, which emulates an intermediate tier typically deployed in data centers for security and access control purposes. In other words, the backend in our implementation supports remote clients by the usage of proxy.

For all experiments, we require the system to commit and execute client requests before responding to clients. Since as observed in \cite{EPaxos_revisited_NSDI2021}, most applications depend on information or confirmation returned by an operation. For example, Redis and ZooKeeper \cite{ZooKeeper_ATC10} return results for both reads and writes. %This is why the throughput and latency of EPaxos reported in our paper are worse than the ones in prior papers \cite{EPaxos_SOSP13,Atlas_Sutra_Eurosys20}. For Giza, we only compare the 

Each experiment is run for 180 seconds, and we collect statistics in the middle of 150 seconds. In our experience, the statistics are quite stable during this period because of the removal of warm-up and cool-down. By default, each object is of 16B. While large objects are common in object storages \cite{Giza_ATC2017,IBM-Cloud-Storage_HotStorage2020}, 16B gives the best performance for EPaxos and Gryff, and in \cite{Facebook_Memcached_workload_sigmetrics12}, Facebook reported the workload of using Memcached as a key-value store, where 40\% of the data is less than 11B. Therefore, we mainly test 16B objects. 

Following prior works \cite{EPaxos_SOSP13,Gryff_Lloyd_NSDI2020,EPaxos_revisited_NSDI2021}, all the systems we test store the data in the main memory, except for two sets of experiments. This choice is reliable as long as the number of concurrent node failures is bounded. We are targeting redundancy across data centers, which are very rare to fail concurrently  \cite{Giza_ATC2017,Atlas_Sutra_Eurosys20,Spanner_OSDI12}. Moreover, there exist solutions that prevent data loss from crashed machines, e.g., persistent memory or disaggregated memory \cite{FORD_Fast22,Intel_persistent_memory,PersistentMemoryStudy_Fast2020,DisaggregatingPersistentMemory_ATC20}. For persistence, we evaluate two approaches: (i) writing to the disk using Redis in Section \ref{s:erasure}, and (ii) writing to a file using Go's OS package in Section \ref{s:eval-durability}. %The m510 nodes are equipped with 256 GB NVMe flash storage. 

Two operations are said to be conflicting with each other if they are targeting the same object (or same key) \cite{EPaxos_SOSP13,Gryff_Lloyd_NSDI2020,EPaxos_revisited_NSDI2021,Giza_ATC2017}.  %We use \textit{conflict rate} as a parameter to control the percentage of operations from a single client that target the same object. 
Following the evaluation in \cite{EPaxos_SOSP13,Gryff_Lloyd_NSDI2020,Atlas_Sutra_Eurosys20,Tempo_Alexey_Eurosys21}, we focus on the evaluations with various conflict rates. A conflict rate $p$ denotes that a client chooses the same key with a probability $p$, and some unique key otherwise. Workload with a Zipfian distribution \cite{EPaxos_revisited_NSDI2021} shows a similar pattern.

\subsection{Summary of Our Findings}
To understand whether \name{} performs well under various settings, we aim at answering the following questions:

\begin{itemize}[noitemsep,nolistsep]
\item Does \name{} reduce tail latency under various conflict rates and write ratios? (Section \ref{s:tail-latency})

\item How does the throughput of \name{} compare to the state-of-the-art competitors? (Section \ref{s:throughput})

\item How does object size and write ratio affect the performance of \name{}? (Section \ref{s:throughput} and Section \ref{s:erasure})

\item How does \name{} perform when integrated with the layered design and erasure coding? (Section \ref{s:erasure})

\item How does persistence affect latency? (Section \ref{s:eval-durability})

\item How does \name{} scale when tolerating a smaller number of concurrent failures? (Section \ref{s:scalability})
\end{itemize}

To summarize our findings, under various conflict rates, \name{} has better read and write tail latency than both Gryff and EPaxos. When $n=3$, around 5\%-18\% of reads are faster than Gryff, even though both systems complete reads in 1 RTT. This demonstrates the effectiveness of our read optimization mentioned in Section \ref{s:optimization}. \name{}'s maximum throughput is 0.5x-4.5x greater than Gryff and EPaxos in the context of geo-replication with a write-intensive workload. The write ratio does not have a significant impact on throughput in \name, whereas it impacts Gryff significantly because of its 2-RTT writes. Finally, \name{}'s reads are 12.5\%-17\% faster than \textsc{Tempo} because of the smaller read quorum size.

\subsection{Tail Latency}
\label{s:tail-latency}

\noindent\textbf{\textit{The Case of $n=3$}.}~
First, we examine the tail latency of \name, with a focus on large-scale web hosting. Since the web hosting applications is usually read-heavy \cite{YCSB_SoCC2010,Facebook_TAO_ATC13,Facebook_Memcached_workload_sigmetrics12}, we use the ratio of 94.5\% read operations with various conflict rate. This write ratio is the same as the YCSB-B workload \cite{YCSB_SoCC2010}. 
Figure \ref{fig:6} presents the cumulative distribution functions (CDF) for both read and write latency for clients from three regions (CA, VA, IR) for three different conflict rates with $n = 3$. Top row represents the CDF for reads and bottom row for writes.

\noindent\textbf{1 RTT Reads for \name{} and Gryff.} Both \name{} and Gryff complete reads in 1 RTT, as shown in the top row of Figure \ref{fig:6}. $\sim$66\% of reads complete after 1-RTT of communication with the nearest quorum (a simple majority) between CA and VA, which has latency of 72ms. Clients in IR are closest to the nodes in IR and VA, so 33\% of the reads complete in 1 RTT between IR and VA, which is 88ms. 

\noindent\textbf{Read Optimization of \name{}.} As mentioned in Section \ref{s:optimization}, \name{} exploits the semantics of linearizable object storages to return reads without any communication when there are concurrent readers co-located within the same data center. Depending on when the concurrent reads are invoked, the latencies vary from 0.755ms to 72ms for clients in CA and VA. 

\noindent\textbf{Impact of Instant Execution.}
As identified in \cite{EPaxos_revisited_NSDI2021,Gryff_Lloyd_NSDI2020}, in EPaxos, some operations need to be delayed due to its dependency tracking, which results into a higher latency. In comparison, Gryff and \name{} can execute an operation instantly. 

\noindent\textbf{Impact of Conflict Rate.}
For both Gryff and \name, conflict rate does not affect latency significantly. This is because reads complete in 1 RTT, and writes always complete in 2 RTT in Gryff. With a higher conflict rate, \name{}'s reads have improved latency in the common case, owing to the read optimization. Higher conflict rate implies a higher chance for reads to tag along. With 25\% conflict, \name{}'s writes occasionally need to take 2-RTT to complete.

%\subsection{The Case of $n=5$}
\vspace{5pt}
\noindent\textbf{\textit{The Case of $n=5$}.}~Figure \ref{fig:n=5} reports the cumulative distribution functions of the latency of reads and writes with $n=5$. In this experiment, we use workload consisting of 49.5\% reads and 50.5\%  writes with 25\% conflicts. The write ratio follows from YCSB-A. Roughly equal amount of operations and the higher conflict rate allow us to observe the performance under concurrent operations. 

Again, having more concurrent operations allows \name{} to complete reads faster. Its writes are also faster than Gryff's because of \name{}'s optimistically  fast writes. A faster write also reduces the chance of concurrent access. This is the reason that, compared to Gryff, more reads in \name{} can be completed on the fast path. These results indicate that even with a write-heavy workload,\footnote{YCSB-A is the only write-heavy workload documented in  \cite{YCSB_SoCC2010}.} \name{} has better latency than both EPaxos and Gryff. 

\begin{figure}[h]
	\includegraphics[width=\linewidth]{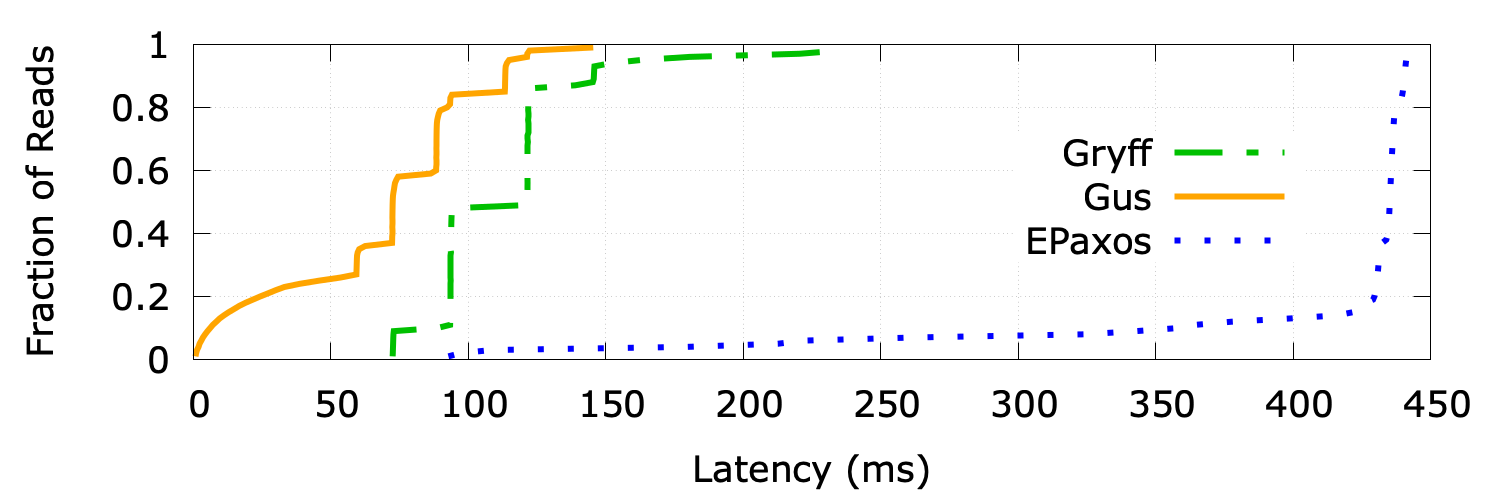}
	\includegraphics[width=\linewidth]{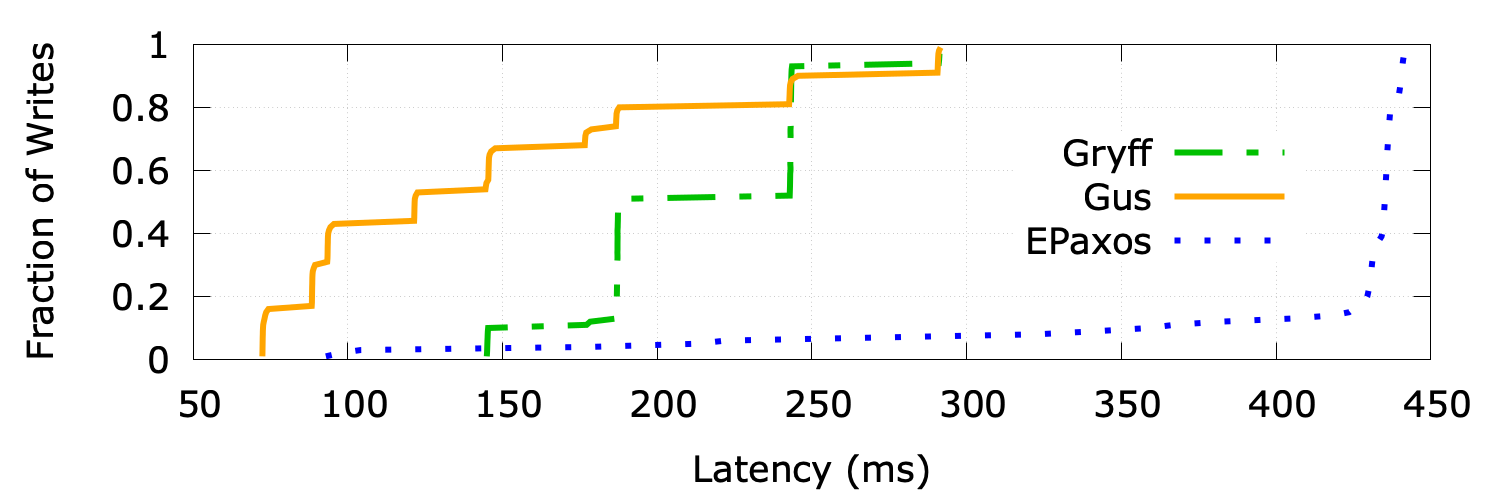}
	\vspace{-20pt}
	\caption{Latency CDF ($n=5$, 50.5\% writes, 25\% conflicts). \name{}'s reads are faster than those of competitors due to the read optimization under high conflict rate and write ratio.}
	\label{fig:n=5}
\end{figure}

\subsection{Throughput vs. Write Ratio in WAN}
\label{s:throughput}

% \vspace{5pt}
% \noindent\textbf{\textit{Throughput vs. Write Ratio in WAN}:}~
We also measure the maximum attainable throughput with various write ratios with $n=3$ (Figure \ref{fig:throughput-writeRatio}). \name{} maintains the highest throughput regardless of the write ratio. Gryff slows down for high write ratios because of its 2-RTT writes. EPaxos has the lowest throughput due to its dependency tracking. %the delayed execution in the wide-area network setting. 

\begin{figure}[!ht]
	\includegraphics[width=\linewidth]{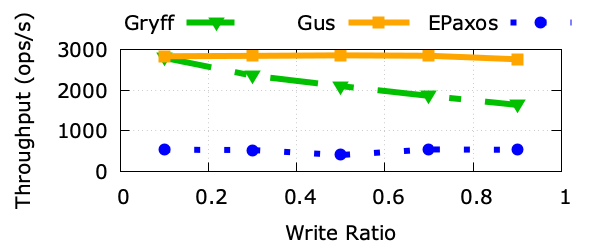}
	\vspace{-20pt}
    \caption{Throughput vs. Write Ratio ($n=3$, 10\% conflicts).}
	%\caption{Maximum attainable throughput vs. Write Ratio ($n=3$, 10\% conflicts).}
	\label{fig:throughput-writeRatio}
\end{figure}

\subsection{Scalability: Layered Design}
\label{s:erasure} 
In this experiment, we test the layered design (Section \ref{s:optimization}), which deploys three metadata servers in the CA, VA, and IR regions, and nine data servers in the five regions identified in Table \ref{t:RTT}, and four other regions not shown in the table for brevity. The maximum RTT between a pair of regions is 243ms, and the lowest is 7ms. To see the impact of two-layered design and isolate the interference from concurrency, we deploy one close-loop client in CA and report the p99.99 latency in Figure \ref{fig:layered} with varying number of data servers. Giza \cite{Giza_ATC2017} was optimized for workload without contention, aligning with our setup. %The objects are persisted to SSDs at the data servers using Redis. 
%For the experiment in this subsection, 
We persist the data and commands to disk using Redis's append-only file feature with a single thread (we invoke \texttt{fsync} at every query).

The number of data servers equals the replication factor, since each server has one copy of data. Reads are scalable when there is no contention, since they take 1 RTT in contacting the metadata servers and retrieving data from the data servers co-located in the same data center. Writes take 2 RTT in this case. \name{} outperforms Giza because of its smaller fast quorum. %Writes' latency improves with 9 data servers, because the two additional servers, compared to the case of 7 data servers, are located near CA; hence, shortening the time spend on the data path. 

\begin{figure}[t]
    \centering
	\includegraphics[width=\linewidth]{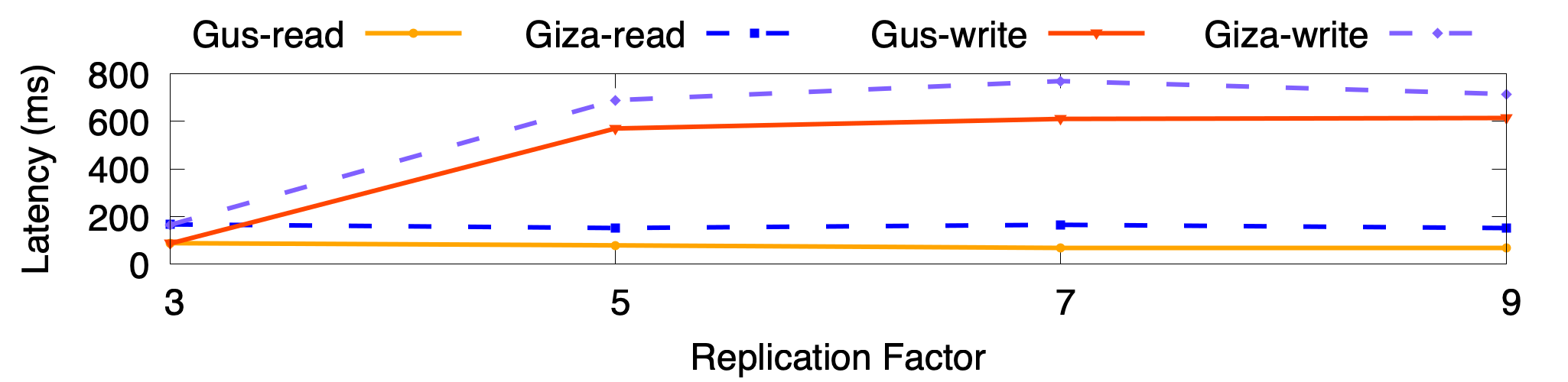}
	\vspace{-10pt}
    \caption{p99.99 Latency vs. Replication Factor (no conflict).}
	%\caption{p99.99 Latency vs. Number of Data Centers/Replication Factor (no conflict).}
	\label{fig:layered}
\end{figure}

\vspace{5pt}
\noindent\textbf{Erasure Coding.}~We integrate erasure coding and layered design in \name{} with three data centers (located in CA, VA and IR). 
%for storing commands.
%Since Giza is optimized for workloads with no contention \cite{Giza_ATC2017}, we deploy one close-loop client in the CA region. 
Figure \ref{fig:EC} reports p50 and p99 read latency. Erasure coding is more cost-effective for larger data in terms of the tradeoff between latency and saving in network bandwidth and storage; hence, we vary data size from 4KB to 4MB. Under the no-contention workload, both systems take 2-RTT (1 RTT to metadata server and 1 RTT to data server). \name{} has better latency due to its smaller fast quorum; its write takes roughly 144ms (2*72), whereas Giza's write takes roughly 223ms (72+151). We observe a similar pattern for write latency. %When we increase the data size to 40MB (not shown in the figure), the latency benefit of \name{} shrinks. \name's latency improves by $\sim$6\% against Giza. This is mainly because for larger data, the time spent on coding/decoding and Redis starts dominating.

\begin{figure}[t]
    \centering
	\includegraphics[width=\linewidth]{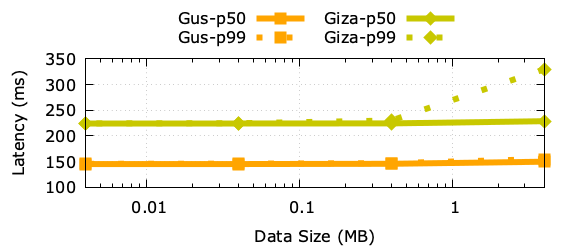}
	\vspace{-10pt}
	\caption{Read Latency vs. Data Size ($n=3$, no conflict). Data size ranges from 4KB to 4MB.}
	\label{fig:EC}
\end{figure}

\subsection{Persistence After Crash}
\label{s:eval-durability}

For persistence, we log state change to an SSD disk before sending  acknowledgement. This experiment uses the same configuration as in Figure \ref{fig-sub:6c}. In EPaxos, nodes log  synchronously to an SSD-backed file, whereas in Gryff and Gus, nodes log their state change only for incoming writes; hence, we only report the latency for writes in Figure \ref{fig:durable}. All the systems are I/O bound, but EPaxos's dependency tracking makes the tail latency increase by  $\sim$600ms, whereas Gryff and \name{} increase by 280-300ms. Even with persistent writes, \name{} still has better tail latency because of its 1-RTT fast path. 

\begin{figure}[t]
	\includegraphics[width=\linewidth]{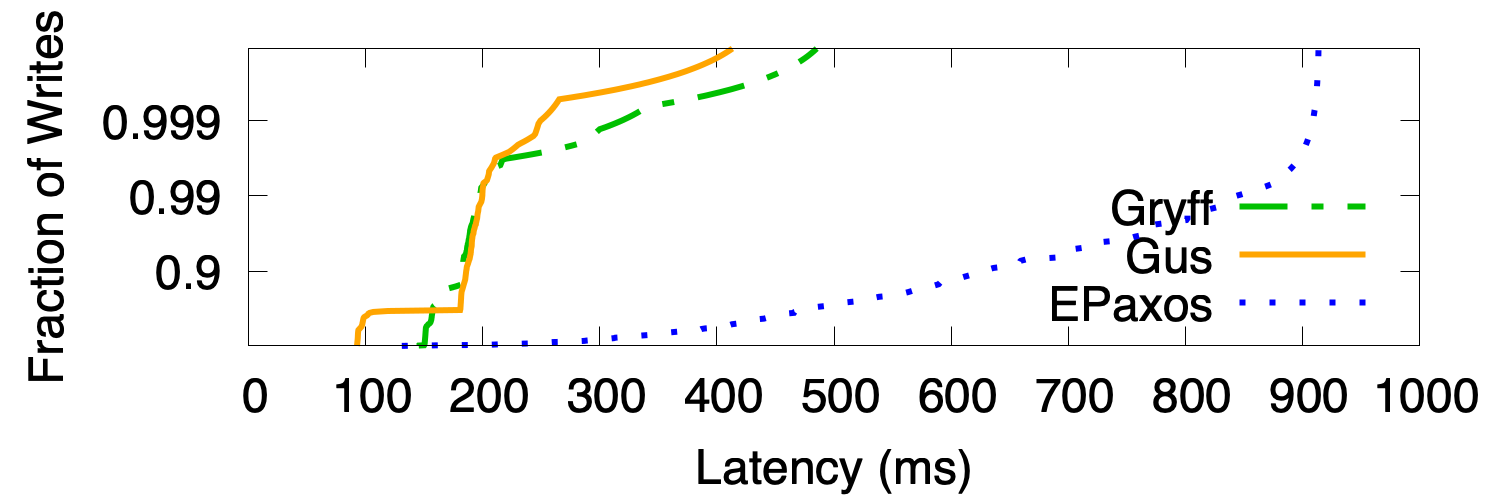}
	\vspace{-20pt}
	\caption{Log-scale Latency CDF with Persistent Writes ($n=3$, 94.5\% reads, %2\% (top) and 
	25\% conflicts).}
	\label{fig:durable}
\end{figure}

\subsection{Scalability: Relaxed Resilience}
\label{s:scalability}

%\name{} can increase scalability by tolerating a smaller number of concurrent failures (Section \ref{s:optimization-scalability}). 
In Figure \ref{fig:scale}, we compare \name{} against \textsc{Tempo} \cite{Tempo_Alexey_Eurosys21} with $n=5, 7, 9$. Both systems tolerate 2  concurrent failures in all three scenarios. To avoid cluttering the plot, we omit the results of EPaxos, Gryff, Flexible Paxos \cite{FlexiblePaxos_Heidi_OPODIS2016}, and \textsc{Atlas} \cite{Atlas_Sutra_Eurosys20}, because they generally have higher tail latency, as also observed in \cite{Tempo_Alexey_Eurosys21,EPaxos_revisited_NSDI2021}. %Since \textsc{Tempo} is a general SMR system that does not distinguish reads and writes, we plot the same latency CDF on both figures. 
%In our setup, the fast path to the closet fast quorum of \name{} takes (i) 72-145ms for both quorums with $n=5$, (ii) 88-170ms for read quorum with $n=7$, (iii) 93-193ms for write quorum with $n=7$, (iv) 78-170ms for read quorum with $n=9$, and (v) 92-209ms for write quorum with $n=9$.

\begin{figure}[t]
	\includegraphics[width=\linewidth]{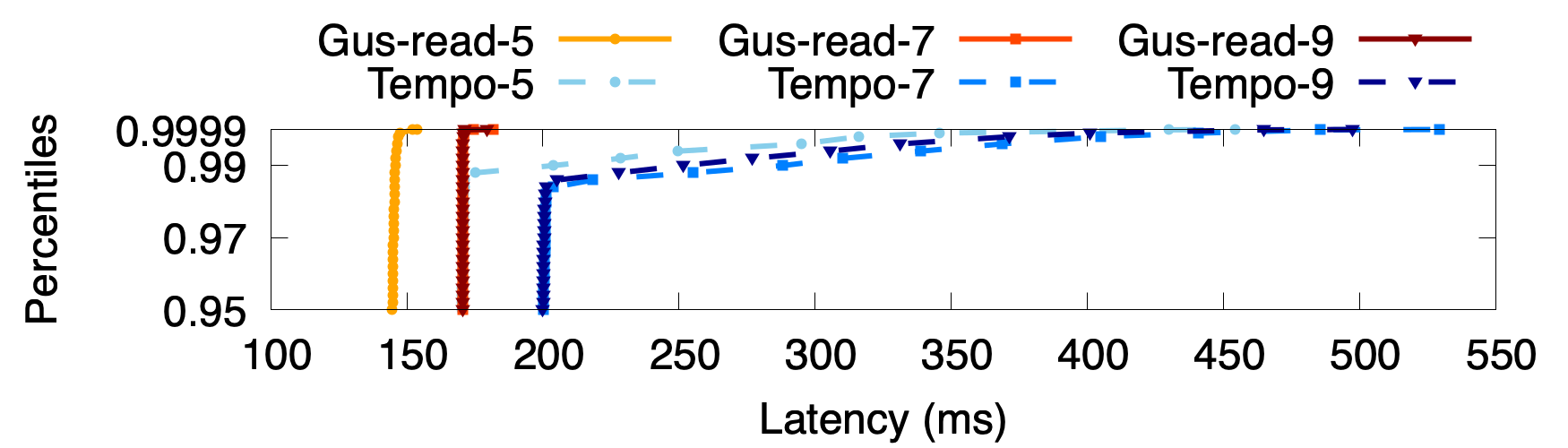}
	\includegraphics[width=\linewidth]{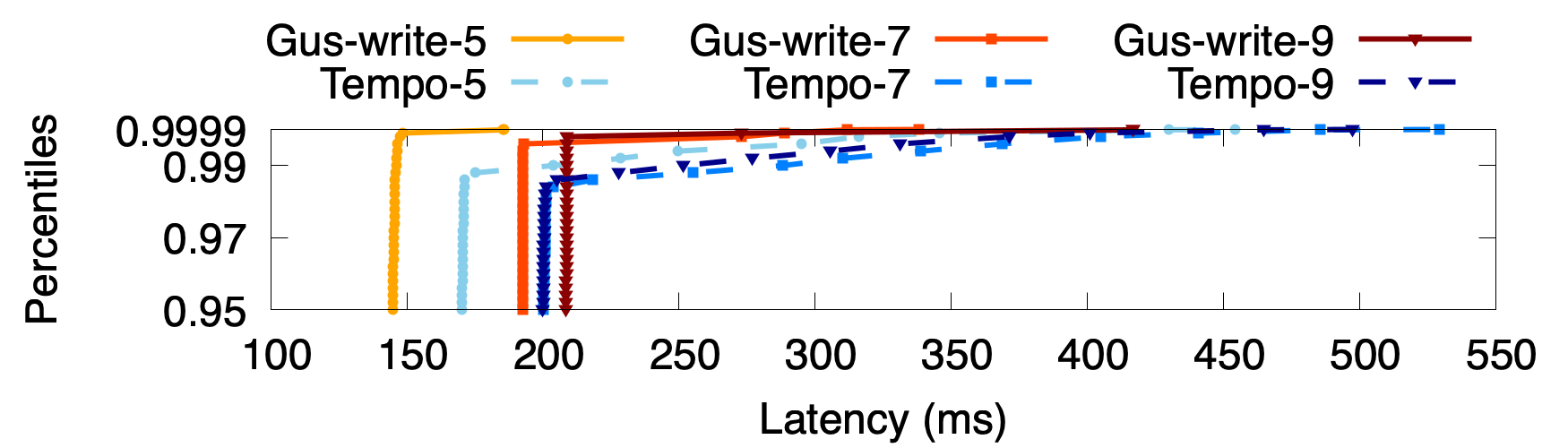}
	\vspace{-20pt}
	\caption{Log-scale Latency CDF with Scalability ($n=5, 7, 9$; 94.5\% reads; 2\% conflicts).}
	\label{fig:scale}
\end{figure}

\name{} has better tail latency for reads because of its smaller read quorum (see Table \ref{t:quorum-size}). For example, when $n=5$, \name{}'s fast path to the closet fast quorum takes 72-145ms and \textsc{Tempo}'s takes 93-162ms. In general, \textsc{Tempo} has better latency for writes when $n=9$, because its quorum is 1 less than \name{}'s write quorum. Occasionally, \textsc{Tempo} needs to wait for timestamps to becomes stable to execute an operation. % (no instant execution property). 
This is mainly the reason that \name{} outperforms \textsc{Tempo} when we consider p99 or above latency for writes. 
	% \section{Conclusion}

% We present \name, a MWMR atomic register with optimal resilience and optimistically fast reads and writes for $n \leq 5$. We show that it is impossible to achieve these two properties when $n > 5$. We experimentally show that \name{} provides lower tail latency than practical systems such as EPaxos, Gryff, Giza, and \textsc{Tempo} under various workloads. 

% \noindent\textbf{Acknowledgement}: 
\section*{Acknowledgement}
Authors would like to thank the anonymous reviewers for the constructive comments. Authors would also like to thank Vitor Enes for helping the evaluation on \textsc{Tempo}. This material is based upon work partially supported by the National Science Foundation under Grant
CNS-2045976 and CNS-1816487.

%We present the design, implementation and evaluation of \name, an object storage system that provides high performance, geo-redundancy, strong consistency, and low tail latency. We achieve these goals by ensuring leaderless design, optimal commit latency, and instant execution. Our evaluation study confirms that \name{} outperforms state-of-the-art competitors. 

\newpage
	
	%\bibliographystyle{plain}
	%\bibliography{\jobname}
	%\balance
	\bibliographystyle{abbrv}
	\bibliography{references,DS_system,Tseng}
	%%%%%%%%%%%%%%%%%%%%%%%%%%%%%%%%%%%%%%%%%%%%%%%%%%%%%%%%%%%%%%%%%%%%%%%%%%%%%%%%
	
	\appendix
	
\newpage

\section{Correctness Proof}
\label{app:proof}

\begin{definition}[Effective Operations]
A read operation is called \textit{effective} if the node invoking the read does not crash while executing it.  A write operation is called \textit{effective} if either the node invoking the write does not crash, or if its value was returned by an effective read.
\end{definition}

\begin{definition}[Committed Write]
A write is committed if a majority of nodes have its value stored in $\storage$.
\end{definition}

An effective write must be committed due to the usage of $staleTag$ and $\commitWrite$ messages. 

\begin{lemma}
\label{lemma:commit-write}
Only the value from a committed write can be returned by a read operation.
\end{lemma}

\begin{proof}
This follows from the definition of a committed write, the usage of $\view$, and the design that a read operation requires a quorum received the same value-tag pair.
\end{proof}

\begin{definition}[Tag for Committed Write]
The tag for a committed write operation is the tag $( \currentTime_i, i)$ associated with the value it writes, where $\currentTime_i$ is defined at Line 3 for writes that take the fast path and Line 10 for writes that take the slow path.
\end{definition}

\begin{lemma}
Tag for any committed write is unique.
\end{lemma}

\begin{proof}
This is because of quorum intersection. During the process when the writer node is still completing the write operation $w$, only two other nodes may serve a read operation. That is, either both of these nodes see no conflict (so the tag won't be updated), or one of the nodes sees conflict, then a read won't return the value in $w$.
\end{proof}

\begin{definition}[Tag for Write]
The tag for a write operation is the tag $( \currentTime_i, i)$ associated with the value it returns, $tag$ at Line 17.
\end{definition}

\begin{lemma}[Unique Write Tag]
\label{lemma:unique-write-tag}
Let $w_1$ and $w_2$ be two effective write operations with tag $(time_1, id_1)$ and $(time_2, id_2)$, respectively. If $w_1 \neq w_2$, then $(time_1, id_1) \neq (time_2, id_2)$.
\end{lemma}

\begin{proof}
If $id_1 = id_2$, then $time_1 \neq time_2$ because of line 2. Otherwise, these two tags won't be the same.
\end{proof}

\begin{lemma}[Progress of Tag]
\label{lemma:op1-before-op2}
Let $op_1$ and $op_2$ be two effective operations with tag $(time_1, id_1)$ and $(time_2, id_2)$, respectively. Suppose $op_1$ terminates before $op_2$ starts. Then we have:

\begin{itemize}
    \item If $op_1$ is a read or a write operation and $op_2$ is a read operation, then $(time_1, id_1) \leq (time_2, id_2)$.
    
    \item If $op_1$ is a read or a write operation and $op_2$ is a write operation, then $(time_1, id_1) < (time_2, id_2)$.
\end{itemize}
\end{lemma}

\begin{proof}
First one is because of quorum intersection.

Second one is because of line 2 and line 10.
\end{proof}

\begin{lemma}[Safety]
There is a total order $S$ on all the effective operations such that (i) $S$ respects the real-time occurrence order for the effective operations; and (ii) any effective read operation obtains the value written by the last write operation that precedes it in $S$.
\end{lemma}

\begin{proof}
Consider only the effective operations. Define the total order $S$ as follows:

\begin{itemize}
    \item Order effective operations according to their tags.
    
    \item If a read and a write have the same tag, the write is ordered in $S$ before the read.
    
    \item If two reads have the same timestamp, the one that starts first is ordered in $S$ before the other one. 
\end{itemize}
Note that by Lemma \ref{lemma:unique-write-tag}, all the writes are totally ordered by their tag. It follows that if two operations have the same tag, one of them is necessarily a read operation.

Given this ordering $S$, we show that is is a linearization of the execution.

\begin{itemize}
    \item Let $op_1$ and $op_2$ be two operations with tag $(time_1, id_1)$ and $(time_2, id_2)$, respectively. Suppose $op_1$ terminates before $op_2$ starts. By Lemma \ref{lemma:op1-before-op2}, we have $(time_1, id_1) \leq (time_2, id_2)$ if $op_2$ is a read operation and $(time_1, id_1) < (time_2, id_2)$ if $op_2$ is a write operation. By the construction of $S$, $op_1$ is ordered before $op_2$.
    
    \item Let $read$ be a read operation that returns a value $value$ with the tag $(time, id)$. By construction of the algorithm, the value is written by node $j$ after it has computed $time$. 
\end{itemize}
\end{proof}
	
\end{document}